\documentclass[11pt, letterpaper]{article}

\usepackage[utf8]{inputenc}
\usepackage[margin=1in, top=1.3in, headheight=65pt, headsep=20pt]{geometry} 
\usepackage{graphicx}
\usepackage{xcolor}
\usepackage{fancyhdr}
\usepackage{titlesec}
\usepackage{hyperref}
\usepackage{lipsum} 
\usepackage[T1]{fontenc}   

\usepackage{amsmath}
\usepackage{amsfonts}
\usepackage{amssymb}
\usepackage{newpxtext,newpxmath} 
\usepackage{microtype}      
\usepackage{multirow}
\usepackage{mathrsfs}
\usepackage{caption}
\usepackage{tabularx}
\usepackage{makecell}

\usepackage{amsthm}

\usepackage{wrapfig}
\usepackage{bm}
\newtheorem{theorem}{Theorem}
\usepackage{enumitem}
\usepackage{booktabs} 

\newcommand{\modelname}{S\textsuperscript{2}-Net}
\usepackage[authoryear,round]{natbib}

\definecolor{CarolinaBlue}{HTML}{4B9CD3}
\definecolor{NavyBlue}{HTML}{13294B}
\hypersetup{
    colorlinks=true,
    linkcolor=NavyBlue,
    filecolor=magenta,      
    urlcolor=CarolinaBlue,
    citecolor=CarolinaBlue
}

\renewcommand{\headrulewidth}{0.8pt}
\renewcommand{\headrule}{\hbox to\headwidth{\color{CarolinaBlue}\leaders\hrule height \headrulewidth\hfill}}

\titleformat{\section}
  {\normalfont\Large\bfseries\color{NavyBlue}}{\thesection}{1em}{}[{\color{CarolinaBlue}\titlerule[0.8pt]}]
  
\titleformat{\subsection}
  {\normalfont\large\bfseries\color{NavyBlue}}{\thesubsection}{1em}{}

\begin{document}

\thispagestyle{fancy}

\begin{center}
    \vspace*{0.5em}
    {\LARGE \textbf{\textcolor{NavyBlue}{From Cortical Synchronous Rhythm to Brain Inspired Learning Mechanism: An Oscillatory Spiking Neural Network with Time-Delayed Coordination}}} \\[1em]
    {\large Tingting Dan and Guorong Wu} \\[0.5em]
    {\small \textit{Department of Psychiatry \& Department of Computer Science}} \\
    {\small \textit{University of North Carolina at Chapel Hill}}
\end{center}
\vspace{1em}

\begin{abstract}
\noindent Human cognition emerges from coordinated spiking dynamics in distributed neural circuits, where information is encoded via both firing rates and precise spike timing determined by brain rhythms. 
Inspired by this notion, we propose a brain-inspired learning primitive in which cognition-level neural synchrony emerges through iterative \textit{bottom-up} and \textit{top-down} interactions between micro-scale dynamics of spiking neurons and a macro-scale mechanism of oscillatory synchronization. 
Specifically, we model each parcel (e.g., a cortical region or an image pixel) in the target system as a spiking neuron embedded in a predefined connectivity scaffold. 
Low-level information is encoded in a spatiotemporal domain, where neurons are selectively grouped and fire spontaneously over time through self-organized dynamics. 
In the bottom-up route, oscillatory synchronization is formed from past spiking activity accumulated over a finite memory window. 
Since brain dynamics operate in a regime of partial and transient synchronization rather than global phase locking, we model oscillatory coordination using a time-delayed synchronization formulation, which enables a top-down modulation of heterogeneous neural spiking for a large-scale distributed system. 
Together, we devise a spiking-by-synchronization neural network (\modelname{}) that uses rhythmic timing as a control mechanism for efficient information processing. 
Promising results have been achieved across a broad range of tasks, including neural activity decoding, energy-efficient signal processing, temporal binding and semantic reasoning.
\end{abstract}

\section{Introduction}
\label{intro}
Despite empirical success, modern machine learning models struggle with dynamic feature binding and temporal adaptation, with existing object-centric solutions often lacking scalability \citep{greff2020binding, locatello2020object}. In contrast, biological systems achieve flexible, energy-efficient computation through rhythmic synchronization that transiently binds relevant features \citep{fries2005mechanism}. Taken together, these principles motivate our interest in developing a brain-inspired architecture that utilizes rhythmic modulation for efficient computation and adaptive reasoning.

\begin{figure}[t]
\centering
\begin{minipage}{0.75\linewidth}
    \centering
\includegraphics[width=1\linewidth, trim=0 0 0 0.1in, clip]{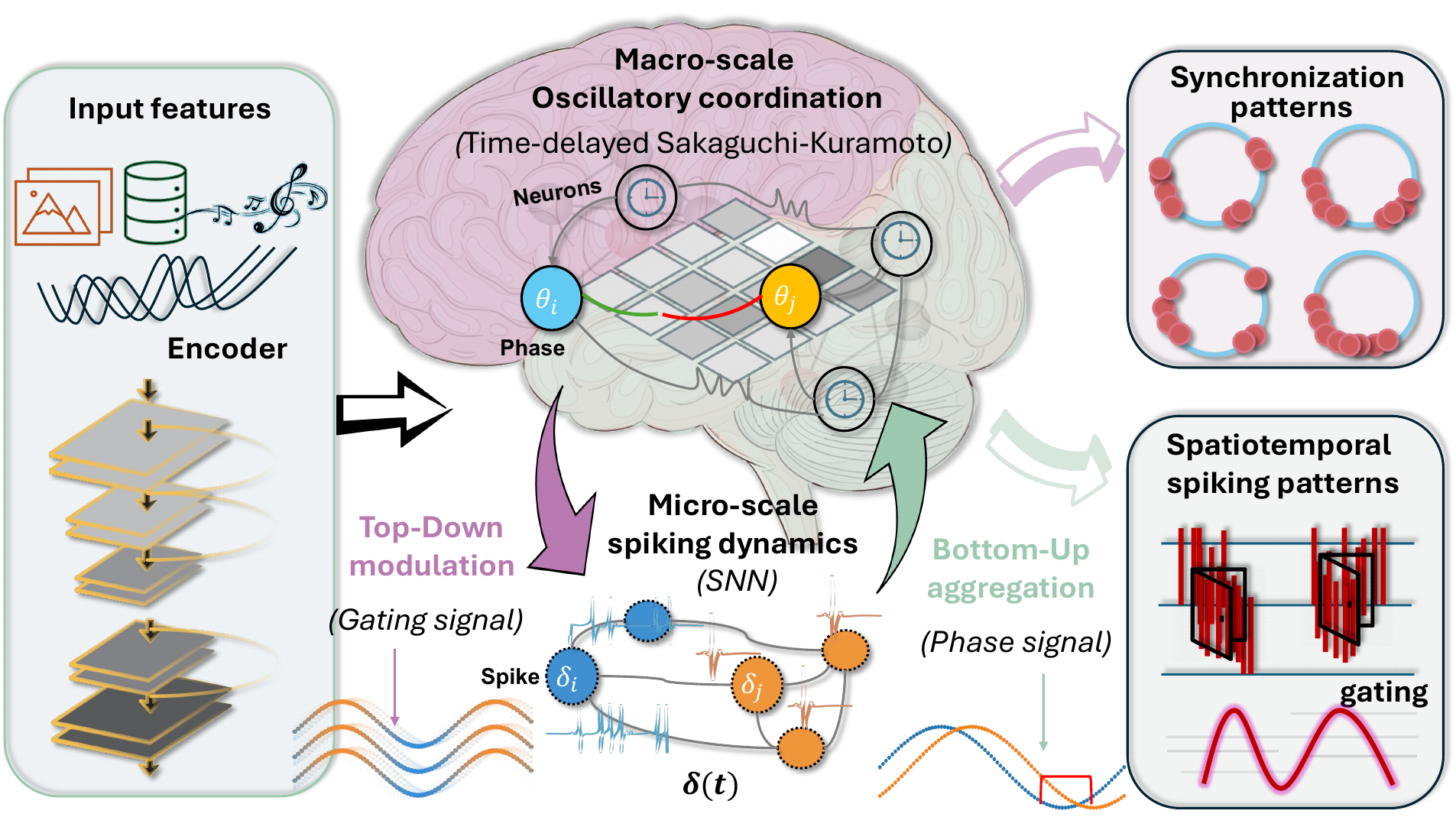}
\end{minipage}
\begin{minipage}{0.24\linewidth}
   \caption{\small {Overview of S\textsuperscript{2}-Net.} The framework integrates a macro-scale oscillatory coordination module with a micro-scale spiking dynamics module. Top-down modulation rhythmically gates neural firing, while bottom-up aggregation feeds phase signals back to the oscillatory layer. This reciprocal interaction enables self-organized synchronization and flexible neural representations.}
    \label{fig:overview}
\end{minipage}
\vskip -0.2in 
\end{figure}

\textbf{Relevant works.} A broad class of brain-inspired deep models is built on spiking neural networks (SNNs) and neuromorphic computing, which replace continuous activations with event-driven spikes to capture temporal dynamics and improve energy efficiency \citep{maass1997networks,roy2019towards}. Recent advances in surrogate gradient learning have enabled SNNs to scale to vision and temporal tasks \citep{neftci2019surrogate}. However, current SNNs have not yet matched the efficiency and robustness of biological neural systems on complex multiscale dynamical tasks. To address this limitation, Rhythm-SNN \citep{yan2025efficient} leverages heterogeneous oscillatory modulation to regulate spiking activity, enabling periodic activation at multiple frequencies. 
Another line of work draws inspiration from iterative inference and binding mechanisms in the brain. Predictive coding models introduce recurrent bottom-up and top-down updates to approximate inference in hierarchical architectures \citep{rao1999predictive}. Meanwhile, object-centric and slot-based models explicitly address the binding problem by grouping distributed features into object-level representations \citep{greff2020binding,locatello2020object}. In the reminiscence of oscillatory neural networks \citep{todri2024computing}, Kuramoto model \citep{kuramoto2005self} is introduced to compress feature representations for the images \citep{miyato2025artificial} and graphs \citep{dan2025explore} via the mechanism of oscillatory synchronization. 

Although the Kuramoto model provides a powerful analytic framework for studying synchronization and has shown promise in brain-inspired machine learning, its standard formulation assumes instantaneous coupling and therefore favors zero-lag global synchronization. Such rigid phase locking has limitations from a perspective of dynamical system, as excessive synchrony collapses representational diversity and limits flexible information processing \citep{singer1999neuronal}. In contrast, empirical evidence indicates that brain dynamics operate in a regime of partial and transient synchronization, where coordinated activity emerges selectively and dissolves as task demands change \citep{varela2001brainweb}. Introducing time delays into the Kuramoto framework relaxes instantaneous coupling and enables metastable, non–zero-lag synchronization patterns \citep{yeung1999time}, providing a more suitable foundation for modeling brain-like coordination and scalable integration with learning systems.

\textbf{Our work.}
In the human visual system, sensory input is progressively encoded through \textit{bottom-up} processing into features such as edges, motion, and color, relayed via the thalamus to the primary visual cortex, and hierarchically integrated along ventral and dorsal pathways to support object recognition and visually guided action \citep{felleman1991distributed}. In parallel, \textit{top-down} feedback from higher-order cortical areas modulates early visual representations based on context, attention, and prior knowledge \citep{rao1999predictive}. The coordinated interaction between bottom-up and top-down pathways, together with neural synchronization and predictive mechanisms, enables coherent, flexible, and adaptive visual perception in dynamic environments.
Inspired by this neuroscience insight, we carve the nature of ``intelligence" through a hierarchical architecture that combines bottom-up sensory signals with top-down predictions, with the design perspective presented in Fig. \ref{fig:overview}.


Our model architecture consists of two interacting layers. At the micro scale, spiking neurons encode information through sparse, temporally precise events, while at the macro scale, neural oscillations regulate communication by selectively synchronizing distributed populations in time \citep{fries2005mechanism}. Specifically, we model each parcel in the system, such as a cortical region or an image pixel, as a spiking neuron, coupled with others in a predefined connectivity scaffold, allowing low-level information to be represented as spatiotemporal spike patterns that self-organize over time. \textit{Bottom-up} signals are aggregated by forming oscillatory synchronization from recent spiking history within a finite memory window, capturing how local activity shapes system-level coordination.
At the macro scale, the signals from the bottom layer are neither static nor globally synchronized. Instead, we model oscillatory coordination using a time-delayed synchronization formulation \citep{yeung1999time}, which relaxes instantaneous coupling and supports metastable, non-zero-lag synchrony. This delayed coordination mechanism provides a top-down control signal that modulates heterogeneous spiking activity across large-scale distributed systems. Together, these components give rise to a spiking-by-synchronization neural network (\modelname{}), in which rhythmic timing functions as an explicit control mechanism for scalable representation binding and efficient information processing.

We first validate such a mechanism on human fMRI data where our \modelname{} accurately predicts the dynamic behavior from the underlying large-scale distributed system of human brain. In addition to the success in replicating neural activities for the human brain, we show that our proposed learning mechanism provides performance improvements across a wide spectrum of tasks, including signal denoising, temporal binding from visual inputs. In a nutshell, these empirical results highlight the importance of rethinking neural representations in a bio-plausible manner, and in particular demonstrate the potential of brain-inspired computing.

\textbf{Our major contributions} have threefold: 
\begin{itemize}[noitemsep, topsep=0pt]
    \item We unify SNN architectures and neural modulation within an interactive bottom-up/top-down framework, inspired by the neuroscience insight that cognition emerges from the interplay between precise spike-based information flow and synchronization-driven temporal coordination.
    \item By taking time delays into account, we relax instantaneous coupling and enable metastable coordination and spatially structured dynamics, which ratchet the gear of brain-like computing another notch forward.
    \item We investigate brain-inspired reasoning grounded in spatiotemporal neuronal spiking patterns, establishing a new connection between cognitive neuroscience principles and modern artificial intelligence (AI). 
\end{itemize}

\section{Method}
\label{sec:method}

\subsection{Network Architecture}\label{subsec:architecture} 
Real-world data often exhibit structured regularities, such as recurrent patterns in audio, scale-invariant self-similarity in natural images, and self-organized functional fluctuations in the human brain. Despite domain differences, these phenomena can be viewed as arising from interacting dynamical systems that exhibit oscillatory modes and structured coordination. This perspective motivates incorporating oscillatory and synchronization-based inductive biases for learning representations of complex data.

\subsubsection{Model input and network components}
Following this notion, we propose a brain-inspired neural network architecture \modelname{}, with a couple of innovative components. 
\textit{First}, we move beyond conventional data representations such as vectors or images. We instead model a generic sensory input $X \in \mathbb{R}^{T \times N}$ as a dynamical system composed of $N$ interacting entities (e.g., anatomical regions, frequency channels, image pixels or feature dimensions ), each evolving as a time series over $T$ time steps. Thus, unlike conventional artificial neural network (ANN) architectures that are explicitly tailored to specific data modalities (e.g., CNNs for images and RNNs for sequential data), our \modelname{} is inherently agnostic to the underlying structure of the input data.
Depending on the data domain, $X$ may represent a time-series of neural activity, a sequence of audio frames, or a static image encoded as a temporal sequence. The interactions among entities are encoded by a weighted graph $A \in \mathbb{R}^{N \times N}$, instantiated from biological structural connectivity, spatial proximity, or inferred from data.
Take functional neuroimages of human brain as an example, $X_i{(t)}$ denotes the BOLD (blood oxygenation level–dependent) activity of brain region $i$ at time $t$, and $A$ represents the inter-regional coupling strength based on fiber counts.

\begin{figure}
\centering
\includegraphics[width=0.99\linewidth, trim=20 0 0 0, clip]{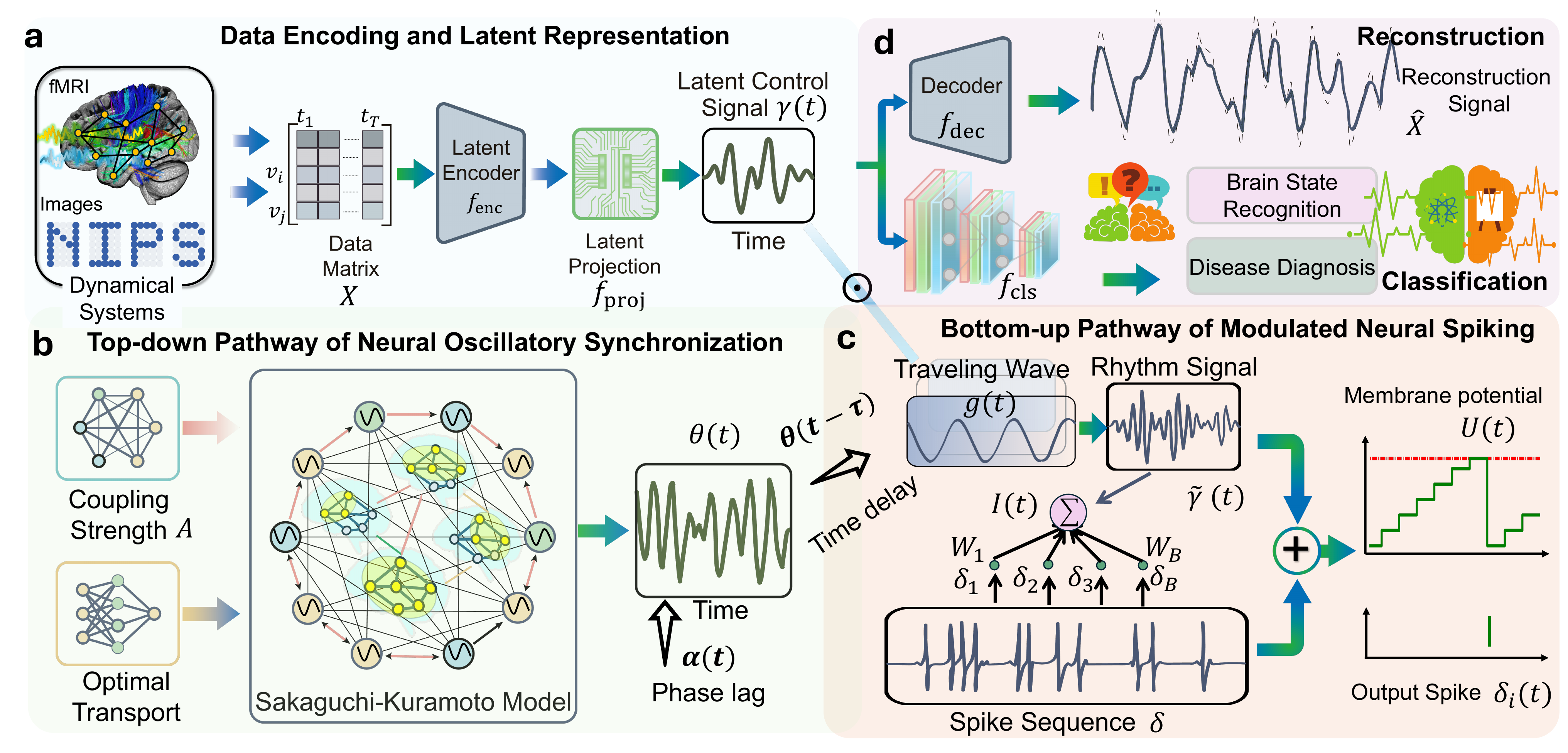}
\vskip -0.10in
\caption{\small \modelname{} architecture. {(a)} Encoding inputs to latent signal $\gamma(t)$. {(b)} Sakaguchi-Kuramoto model governs phase states $\theta(t)$. {(c)} Delayed phases $\theta(t-\tau)$ gate $\gamma(t)$ to modulate membrane potential $U(t)$ and drive SNN spike output. {(d)} Decoding and classification for downstream tasks. }
\label{fig:flow}
\end{figure}
\textit{Second}, the model backbone is an interactive spiking–synchronization ($S^2$) architecture in which neural synchrony emerges through iterative bottom-up and top-down interactions between micro-scale spiking dynamics (by SNN) and a macro-scale mechanism of oscillatory coordination (by Sakaguchi-Kuramoto model).
As shown in Fig. \ref{fig:overview}-middle, we deploy paired \textbf{spiking neurons} $\mathbb{S}=\{\mathscr{S}_1,\mathscr{S}_2,...,\mathscr{S}_N\}$ at the bottom layer and corresponding \textbf{gating neurons} $\mathbb{G}=\{\mathscr{G}_1,\mathscr{G}_2,...,\mathscr{G}_N\}$ at the upper layer, each pair is associated with the data entity $X_i$. The spiking activity on $\mathscr{S}_i$ is represented as a binary time series $\delta_i(t)\in \{0,1\}$, where 0 and 1 denote the inactive and active states, respectively. In parallel, each gating neuron $\mathscr{G}_i$ maintains a phase state $\theta_i(t)$. 
Conditioned on the current spiking configuration $\{\delta_i(t)\}$, the collective behavior of neuron firing gives rise to the phase lag $\alpha_{ij}$ between $\mathscr{G}_i$ and $\mathscr{G}_j$ during neural oscillatory synchronization.   
Meanwhile, we introduce a time delay mechanism among gating neurons to induce a rhythm signal for modulating spiking patterns across spiking population $\mathbb{S} $. 
Together, phase lag and time delay provide a principled mechanism for coupling micro-scale spiking dynamics with macro-scale oscillatory coordination with an enriched spectrum of system dynamics, enabling partial synchronization, directional interactions, and metastable transitions that are difficult to capture with delay-free coupling.

\subsubsection{Learning and inference in \modelname{}}
The fitting of model parameters involves the following learning components.

\textbf{Latent projection of control signals.} To extract coordination-relevant structure from high-dimensional inputs, we employ a two-stage encoding strategy to map raw signals into a latent driving representation (Fig. \ref{fig:flow}a):
\begin{equation}
\gamma(t) = f_{\mathrm{proj}}(z(t)) = f_{\mathrm{proj}}(f_{\mathrm{enc}}(X(t))),  \gamma(t) \in \mathbb{R}^{N}.
\label{eq:proj}
\end{equation}
\textit{First}, we deploy an encoder $f_{\mathrm{enc}}$ to remove external noise from the input signal $X(t)$.
\textit{Second}, we project the noise-free signals $z(t)$ to a time-varying \textit{control signal} $\gamma(t)$, which modulates the neural oscillatory synchronization dynamics. Specifically, $\gamma(t)$ acts on the vector-valued \textit{phase state} $\Theta (t)=\{\theta_i(t)|\theta_i(t)\in \mathbb{R}^D, i=1,...,N\}$ and steers its evolution under the Kuramoto dynamics. This design is inspired by the vector-valued Kuramoto dynamics \citep{miyato2025artificial, dan2025explore}, enabling the system to "bind" desired task-relevant patterns through synchronization.

\textbf{Infer model parameters from data.} 
The paired design of gating neurons and spiking neurons is intended to mimic key principles of information processing in the human brain, where cognition emerges from the continuous interaction between bottom-up sensory-driven signals and top-down modulatory control. 
Under the hood of bottom-up and top-down interaction, we introduce a general-purpose coordination mechanism that involves: {(i)} controlled neural oscillatory synchronization with phase delay for temporal binding and segregation, {(ii)} rhythmic gating for temporal multiplexing from traveling waves, and {(iii)} rhythm-modulated spiking neural computation for downstream decoding. The working mechanism of each component is explained in Sec. \ref{subsec:learn}

\textbf{Link to downstream tasks.} We iteratively alternate between top-down and bottom-up processes until the neural spiking patterns converge to optimal performance on downstream tasks. Meanwhile, the dynamics of neural firing and oscillatory synchronization in \modelname{} afford decent model interpretability, enabling principled analysis of machine intelligence in reasoning and problem solving.

\subsection{Learning Mechanisms of Synchronized Spiking}\label{subsec:learn} 
\subsubsection{Top-down pathway of neural oscillatory synchronization} \label{subsec:top_down}
The mechanistic role of the top-down pathway (Fig. \ref{fig:flow}b) is to update the phase state of each gating neuron in the upper-layer and generate the rhythm signals that selectively gate bottom-layer spiking neurons as a consequence of oscillatory phase synchronization. We formulate the evolution of vector-valued oscillatory state $\theta_i(t) \in \mathbb{R}^D$ in a topologically constrained Sakaguchi-Kuramoto system, which incorporates intrinsic frequencies, sensory forcing, and phase lags emerging from the bottom-layer spiking neurons.

\textbf{Coupled vector Kuramoto dynamics with time-dependent phase lags.}\label{subsec:vector_kuramoto} 
For the $d$-th component in $\theta_i$ ($d \in \{1,\dots,D\}$), the evolution of $\theta_i^d$ are determined by:
\begin{equation}
\dot{\theta}^d_{i}(t) = \omega_{i}^d + \kappa_{i}^d\sin\bigl(\gamma_i(t) - \theta_{i}^d(t)\bigr) + \frac{K}{N} \sum_{j=1}^{N} A_{ij} \sin \left( \theta_{j}^d(t) - \theta_{i}^d(t) - \alpha_{ij}(t) \right) \label{eq:vector-kuramoto}
\end{equation}
where $K$ is the global coupling strength, $A_{ij}$ is the coupling strength constraining the interaction range, and $\alpha_{ij}(t)$ characterizes the phase lag between two coupled gating neurons at time $t$. $\omega_{i}^d$ represents the learnable \textit{natural frequency} of the oscillator, capturing the intrinsic characteristic of neural vibration. $\kappa_{i}^d$ controls the strength of the attraction toward the data-specific control pattern $\gamma_i(t)$, ensuring that the dynamical system reacts to the incoming external stimuli.

\textbf{Intuition of oscillatory synchronization in Sakaguchi-Kuramoto model.}
The Sakaguchi Kuramoto model in Eq. (\ref{eq:vector-kuramoto}) shapes oscillatory synchronization by phase-shifted and delayed interactions rather than instantaneous alignment. 
$\gamma_i(t)$ in the second term is an external "push" on the phase state that steers the synchronization pattern toward a desired phase configuration. To that end, the dynamics in Eq. (\ref{eq:vector-kuramoto}) naturally give rise to metastable phase-locked patterns, which induce task-relevant rhythmic gating signals.
The third term in Eq. \ref{eq:vector-kuramoto} introduces a time-varying phase lag $\alpha_{ij}$ that arises from bottom-layer spiking dynamics. In the presence of this phase lag, the dynamics in Eq. \ref{eq:vector-kuramoto} seek to align the phase-shifted oscillations, thereby generating rhythmic signals that selectively trigger bottom-layer spiking neurons.

\textbf{Understand the mechanistic role of top-down pathway through potential energy and Lyapunov stability.}
Under the assumption of symmetric coupling ($A_{ij}=A_{ji}$) and antisymmetric phase lags ($\alpha_{ij}=-\alpha_{ji}$), the dynamics of phase transition characterized by Eq.~\eqref{eq:vector-kuramoto} is governed by the potential energy $V({\theta_i^d}; t)$:
\begin{equation}
\label{eq:lyapunov}
V({\theta_i^d}; t) = \underbrace{-\frac{K}{2N}\sum_{j} A_{ij} \cos(\theta_{j}^d(t) - \theta_{i}^d(t) - \alpha_{ij}(t))}_{V_{\mathrm{sync}}: \text{System dynamics}} \underbrace{- \kappa_{i}^d \cos(\gamma_i(t) - \theta_{i}^d(t)) \vphantom{\sum_{j}}}_{V_{\mathrm{data}}: \text{System control}}
\end{equation}
The term $V_{\mathrm{sync}}$ decomposes the Sakaguchi–Kuramoto dynamics into (i) coherent phase-gradient patterns and (ii) a non-conservative rotational component that induces uniform drift along the energy manifold, converting static phase-locked states into \textit{traveling waves} \citep{bronski2018configurational,ju2014dynamics}.
In parallel, $V_{\mathrm{data}}$ acts as a potential well anchoring the state under the control of $\gamma_i(t)$.
As shown in Appendix~\ref{app:stability}, the potential function $V$ is a global Lyapunov function.
Thus, the dynamics in Eq.~\eqref{eq:vector-kuramoto} satisfy $\dot{\theta}_{i}^d \propto -\partial V / \partial \theta_{i}^d$, implying that the total energy decreases monotonically ($\dot{V} \leq 0$). This guarantees that the neural oscillatory synchronization converges to a stable configuration that optimally balances brain geometry constraints with observational data consistency.

\textbf{Interaction with the bottom-up pathway: Generating rhythm signals for modulating neural spiking.} 
A growing body of neurophysiological evidence shows that cortical and hippocampal activity is organized via traveling oscillatory waves, spanning theta to gamma bands, which are closely linked to perception, attention, working memory, and learning \citep{deco2009key,breakspear2017dynamic}. Following this notion, we further introduce a (predefined or learnable) time delay $\tau$ on each gating neuron to generate the rhythmic gating signal $g(t)$ by binding with the delayed phase state $\theta_i(t-\tau)$:
\begin{equation}
g_i(t)
= \sigma_{\text{gate}}\Bigl(\sin\bigl(\theta_i(t-\tau)\bigr)\Bigr), \quad i=1,...,N,
\label{eq:gating}
\end{equation}
where $\sigma_{\text{gate}}(\cdot)$ scales the sine wave to $[0,1]$. 
The time delay $\tau$ serves as a critical bridge between the top-down gating and bottom-up execution pathways. Building on the oscillatory dynamics in Eq. (\ref{eq:vector-kuramoto}), the delay is functioning as a frequency-dependent gating mechanism that selectively activates specific phase patterns as traveling waves. We exploit this property to dynamically determine the grouping of spiking neurons through the rhythm signal $\tilde{\gamma}_i(t)$: 
\begin{equation}
\tilde{\gamma}_i(t) = g_i(t) \odot \gamma_i(t),
\end{equation}
where `$\odot$' is the Hadamard operator. 
Through this process, static phase separation is converted into temporal multiplexing, enabling selective binding via phase-aligned timing rather than spatial convergence across spiking neurons.

\subsubsection{Bottom-up pathway of modulated neural spiking}
\label{sec:bottom-up}

\textbf{Rhythm-modulated spiking neural network.}
Unlike point-neuron formulations, each spiking neuron comprises multiple dendritic branches with distinct temporal filtering constants, reflecting the compartmentalized integration observed in biological dendrites \citep{london2005dendritic}. This structure enables simultaneous modeling of rapid oscillatory activity and slower contextual dependencies. 
The workflow of \modelname{} consists of the following major steps.

\textit{(1) Spike generation conditioned on rhythm.} The input is an aggregation of $B$ dendritic branches. The raw activation $J_{i,b}(t)$ of the $b$-th branch incorporates the rhythm-modulated signal $\tilde{\gamma}_i(t)$ on $i^{th}$ spiking neuron as:
$
J_{i,b}(t) = {W}_{i,b}^{(\gamma)} \cdot \tilde{\gamma}_i(t) + \sum\nolimits_{j} W_{ij,b}^{(\mathrm{rec})} \delta_j(t-1),$
where ${W}_{i,b}^{(\gamma)}$ projects the oscillatory latent vector onto this selected branch (modulated by $\tilde{\gamma}_i(t)$), and $W_{ij,b}^{(\mathrm{rec})}$ captures recurrent interactions. The branch state $l_{i,b}(t)$ evolves as a low-pass filter with a branch-specific decay rate $\beta_b \in (0,1)$:
$
l_{i,b}(t) = \beta_b l_{i,b}(t-1) + (1-\beta_b) J_{i,b}(t).$
Thus, the total dendritic current $I_i(t)$ is formed by aggregating the filtered outputs of all branches: $I_i(t) = \sum_{b=1}^{B} l_{i,b}(t)$. 

\textit{(2) Modulated membrane potential update.} The somatic membrane potential update is governed by the external rhythmic gate $g_i(t)$ derived in Eq. \ref{eq:gating}. We implement a soft gating mechanism where $g_i(t)$ modulates the effective integration rate of the neuron. The hidden membrane potential $U_i(t)$ is updated via interpolation between the previous state and the new leaky-integrated state $V_i(t)$:
\begin{equation}
V_i(t) = \epsilon U_i(t-1) + I_i(t) - \vartheta \delta_i(t-1), \quad
U_i(t) = (1 - g_i(t)) \cdot U_i(t-1) + g_i(t) \cdot V_i(t),
\label{eq:srnn}
\end{equation}
where $\epsilon$ is the membrane leak factor and $\vartheta$ is the firing threshold. This mechanism allows the rhythm to dynamically ``freeze" ($g \rightarrow 0$) or ``update" ($g \rightarrow 1$) the neural state, preserving information during off-cycles while enabling gradient propagation. The output spike $\delta_i(t)\in \{0,1\}$ is generated as $\delta_i(t) = \Xi(U_i(t) - \vartheta)$, where $\Xi(\cdot)$ denotes a Heaviside step function.

\textit{(3) Spatiotemporal spiking feature aggregation.} We employ a fully connected readout layer to map the population spiking activity $\delta(t)$ to a continuous-valued output $\mathcal{H}(t)$. This layer integrates spatiotemporal spike patterns into a low-dimensional decision variable that is differentiable and amenable to optimization for downstream learning tasks.

\textbf{Interaction with the top-down pathway: Infer phase lag from the temporal structure of spiking neurons.}
Although direct measurement of the phase lag is challenging, we introduce an intuitive recursive strategy to estimate $\alpha_{ij}(t)$, \textit{on-the-fly}, using the temporal history of spiking profiles $\{\delta(t)\}$ and $\{\delta(t-1)\}$ (Fig. \ref{fig:flow}c). Assuming that gating neurons are topologically connected according to the adjacency matrix $A$, we cast this estimation as an entropically regularized Wasserstein optimal transport problem defined on the graph induced by $A$, where the objective minimizes the transport cost ($C_{ij} \propto 1/({A_{ij} + \varepsilon}$) between consecutive spiking distributions subject to marginal constraints. The resulting optimal transport plan is efficiently computed via Sinkhorn iterations, enabling scalable and online inference of $\alpha_{ij}(t)$ through fast iterative matrix–vector scaling \citep{cuturi2013sinkhorn}.

\textbf{Deploy \modelname{} in downstream applications}
We drive feature learning using three downstream applications (Fig. \ref{fig:flow}d): (1) \emph{self-supervised reconstruction} as a regularizer, (2) \emph{dynamic state recognition} for sequence labeling, and (3) \emph{subject-level classification} via temporal pooling. Detailed formulations and loss functions are provided in Appendix~\ref{subsec:model}.  

\subsection{Proof-of-Concept Simulation Design} 
   
\textbf{Experimental setup.} To validate \textit{binding-by-synchrony} and \textit{traveling-wave propagation}, we constructed two synthetic environments (Fig. \ref{fig:sim2}): (a) disjoint squares and (b) nested topological structures\footnote{In computational neuroscience, nested structures are regarded as a canonical benchmark \citep{roelfsema1998object,chen1982topological} for figure-ground segregation. They challenge the system to solve the 'binding problem' by integrating spatially separated segments into a coherent whole.}. Objects are initialized with distinct intensities ($1.0$ vs. $0.5$) to enable feature-driven repulsion via the graph coupling term $A_{ij}$. We applied Gaussian perturbation at varying levels: a baseline $\sigma=0.05$ for (a), and a severe noise regime $\sigma=0.2$ for (b) to test robustness under low signal-to-noise ratio.

\textbf{Discrete object discovery and attention rhythms (Fig. \ref{fig:sim2}a).}
In the disjoint task (two objects shown in Fig. \ref{fig:sim2}a-left), our \modelname{} breaks phase symmetry, driving the system from a disordered transient state into a stable anti-phase regime. Interestingly, the gating dynamics $g(t)$ (Fig. \ref{fig:sim2}a-middle) emerge as smooth, alternating sinusoidal waves. This dynamic creates a "push-pull" mechanism that mathematically replicates biological attention rhythms (e.g., Theta oscillations). The resulting spatiotemporal spike raster (Fig. \ref{fig:sim2}a-right) reveals temporal chunking, where objects are phase-locked to alternating firing windows separated by refractory gaps.

\textbf{Topological complexity and perceptual bistability (Fig. \ref{fig:sim2}b).} Standard coupled oscillators often collapse into in-phase synchronization or a winner-take-all regime where the larger ring suppresses the smaller square (Fig. \ref{fig:sim2}a-left). In contrast, \modelname{} achieves perceptual bistability by leveraging graph-aware traveling waves. As shown in the gating signal map (Fig. \ref{fig:sim2}b-middle), a phase gradient propagates along the ring, binding its spatially disjoint segments into a unified entity. This mechanism ensures balanced strength, despite the size imbalance, the smaller square is not suppressed but competes equally for representation. The spike raster (Fig. \ref{fig:sim2}b-right) reveals a traveling wave topology, where the square fires precisely within the local phase troughs of the ring's propagation, resolving the topological conflict through autonomous temporal multiplexing.

\textbf{Quantitative evaluation and stability analysis.} We further evaluate feature binding using Kuramoto-derived metrics (Fig. \ref{fig:sim}d). The system rapidly transitions from a disordered state to a stable synchronized manifold. Specifically, intra-object synchronization converges to a theoretical maximum ($R_{intra} \to 1.0$) for perfect internal binding, while inter-object overlap collapses ($R_{merge} \to 0$) to ensure clear anti-phase separation. See Appendix \ref{sim_app} for detailed derivations.

\begin{figure}
\centering
\includegraphics[width=1\linewidth, trim=0 0 0in 0in, clip]{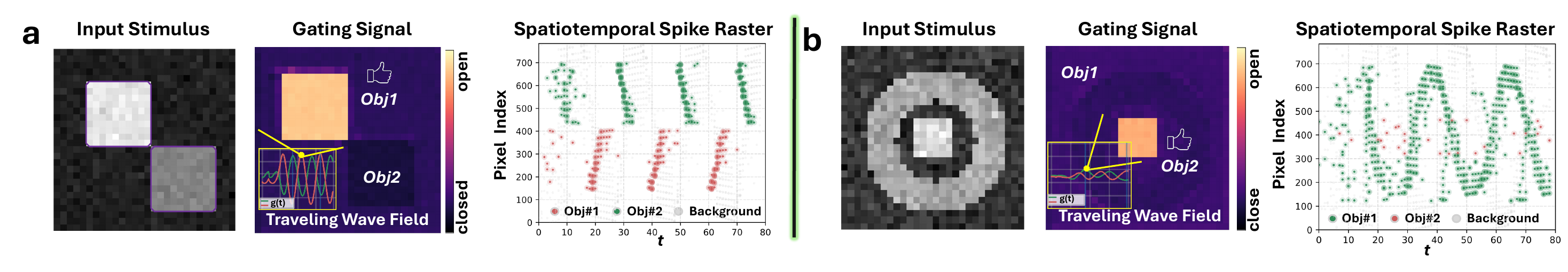}
\vskip -0.13in
\caption{\small Validation of the power of \modelname{} framework. (a) Disjoint Object Discovery.  (b) Nested Topological Binding.
Panels (left to right) show the input stimulus, gating signals $g(t)$ reflecting attention rhythms, and spatiotemporal spike rasters showing wave propagation. }
\label{fig:sim2}
\end{figure}

\section{Experiments}
\label{experiement}
\subsection{Experimental Setup}
\textbf{Comparison methods.} We compare some representative SNN architectures, including:
(i) a feedforward SNN (FFSNN) \citep{wu2019direct};
(ii) a recurrent SNN (SRNN) \citep{zenke2021remarkable};
(iii) the long short-term memory SNN (LSNN) \citep{bellec2018long};
(iv) an adaptive SRNN (ASRNN) \citep{yin2020effective};
(v) a dendritic heterogeneity SNN (DH-SNN) and
(vi) Rhythm-SNN \citep{yan2025efficient}.

\textbf{Human brain dataset.} We evaluate \modelname{} on six large-scale neuroimaging datasets on both task recognition and disease diagnosis settings.
For task recognition, we use HCP-A \citep{bookheimer2019lifespan}, HCP-YA \citep{van2013wu}, and UK Biobank (UKB) \citep{miller2016multimodal} datasets, which include multiple task-based fMRI paradigms with corresponding diffusion MRI for constructing structural brain networks. In addition, we evaluate our method on the Working Memory task (HCP-WM) from HCP-YA, which involves alternating cognitive states within a single fMRI scan. This setting requires time-point–level task state recognition rather than assigning a single label to an entire scan.
This formulation is more challenging, as it demands precise modeling of rapid task transitions, temporal dependencies, and subtle changes in functional brain dynamics.
For disease diagnosis, we adopt ADNI \citep{jack2008alzheimer}, PPMI \citep{marek2011parkinson}, and NIFD \citep{planche2023anatomical} datasets, covering Alzheimer’s disease (AD), Parkinson’s disease (PD), and frontotemporal dementia (FTD), respectively, using resting-state fMRI and diffusion MRI data.
Across all datasets, task recognition is formulated as a multi-class classification problem, whereas disease diagnosis is treated as a binary or multi-class classification task depending on the dataset.
Detailed dataset descriptions and experimental settings are provided in Appendix \ref{dataset} and \ref{details}. For quantitative evaluation, we use accuracy (Acc), precision (Pre), and F1-score (F1) to assess performance, while the quality of signal reconstruction is quantified using mean absolute error (MAE) and  mean squared error (MSE).

\subsection{Experiment Results on Neural Activity}

\textbf{Brain state recognition.} Table \ref{results} (left) presents the quantitative comparison of \modelname{} against state-of-the-art (SOTA) SNN methods on HCP-A, HCP-YA, and UKB. 
The results demonstrate that \modelname{} consistently outperforms all baseline methods across all metrics (*$p<0.05$, paired-$t$-test). 
Even compared to advanced baselines such as DH-SNN and Rhythm-SNN, which incorporate dendritic heterogeneity and neural oscillations respectively. 
This indicates that modeling the coupled interaction between structural and functional information allows for a more robust identification of distinct cognitive states across diverse subject populations. 

\textbf{Brian neural activity decoding.}
We further evaluate the model's capability in fine-grained temporal decoding using HCP-WM dataset. 
As shown in Table \ref{results} (middle), this task is particularly challenging for conventional SNNs. For instance, FFSNN and simple recurrent models achieve significantly lower accuracy due to the complexity of the temporal dynamics. 
However, \modelname{} achieves a remarkable performance leap. 
This suggests that our architecture is uniquely capable of decoding rapid cognitive shifts. Furthermore, we assess the signal reconstruction quality to ensure the model learns meaningful biological representations rather than just optimizing for classification. 
Fig. \ref{fig:recon} shows the reconstruction error on the HCP-WM dataset. 
\modelname{} achieves the lowest MSE and MAE, significantly outperforming the comparative methods (e.g., Rhythm-SNN and DH-SNN). 
This low reconstruction error highlights that \modelname{} preserves the intrinsic topological and temporal structure of the brain signals while performing high-level decoding.

\textbf{Disease diagnosis.}
In the context of neurological disorder diagnosis (ADNI, PPMI, and NIFD), robustness to noise and small sample sizes is critical. 
As reported in Table \ref{results} (right), \modelname{} demonstrates SOTA performance in detecting AD, PD, and FTD. 
While datasets like PPMI and NIFD present high inter-class similarity and limited training samples, our method yields the highest detection accuracy and F1-scores compared to the baselines. 
This superiority suggests that \modelname{} effectively extracts discriminative biomarkers, offering a promising tool for early computer-aided diagnosis.

\begin{table*}[t]
\centering
\begin{minipage}[c]{0.68\textwidth}
    \centering
    \setlength{\tabcolsep}{3pt} 
    \renewcommand{\arraystretch}{1.05}
    \caption{Performance comparison of \modelname{} against SOTA SNN models across human brain datasets via 10-fold cross-validation. `*' denotes significantly outperforming other methods.}
    \label{results}
    \resizebox{\textwidth}{!}{%
    \begin{tabular}{l l c c c c c c c} 
    \toprule
    \multirow{2}{*}{\textbf{Model}} & \multirow{2}{*}{\textbf{Metric}} & \textbf{HCPA} & \textbf{HCP-YA} & \textbf{UKB} & \textbf{HCP-WM} & \textbf{ADNI} & \textbf{PPMI} & \textbf{NIFD} \\
    & & (4 tasks) & (7 tasks) & (2 tasks) & (9 tasks) & (3 states) & (4 states) & (5 states) \\
    \midrule

    \multirow{3}{*}{FFSNN} 
    & Acc & 97.33±0.70 & 95.46±1.51 & 97.93±0.22 & 37.22±0.94 & 91.00±11.74 & 65.08±3.35 & 49.37±1.97 \\
    & Pre & 97.03±0.94 & 95.62±1.50 & 97.94±0.22 & 35.26±0.91 & 94.50±7.14 & 50.11±5.20 & 38.00±7.14 \\
    & F1  & 95.99±1.03 & 95.47±1.50 & 97.93±0.22 & 33.07±0.88
 & 89.27±13.72 & 35.41±6.19 & 29.63±3.40 \\
    \midrule

    \multirow{3}{*}{SRNN} 
    & Acc & 96.67±0.60 & 94.49±1.42 & 97.02±0.16 & 49.74±2.81 & 98.00±6.00 & 66.09±2.56 & 46.42±1.27 \\
    & Pre & 95.53±1.11 & 94.61±1.35 & 97.02±0.16 & 47.66±3.28 & 98.75±3.75 & 45.75±4.66 & 35.18±11.07 \\
    & F1  & 94.89±0.99 & 94.47±1.41 & 97.02±0.16 & 45.73±3.19
 & 97.62±7.14 & 39.84±4.60 & 25.35±4.12 \\
    \midrule

    \multirow{3}{*}{LSNN} 
    & Acc & 96.64±0.42 & 94.68±1.46 & 98.15±0.23 & 42.81±3.45 & 96.67±6.67
 & 66.09±4.21 & 51.40±1.38 \\
    & Pre & 95.77±0.90 & 94.78±1.39 & 98.17±0.23 & 40.00±3.50 & 97.75±4.53 & 47.21±5.01 & 43.34±9.14 \\
    & F1  & 94.81±0.79 & 94.63±1.49 & 98.15±0.23 & 38.60±3.77 & 96.06±7.95
 & 38.46±7.39 & 35.90±3.30 \\
    \midrule

    \multirow{3}{*}{ASRNN} 
    & Acc & 96.46±0.59 & 92.83±1.25 & 98.79±0.08 & 78.07±3.11 & 96.33±7.37 & 64.38±3.78 & 46.06±2.38 \\
    & Pre & 95.54±1.40 & 92.99±1.24 & 98.81±0.08 & 78.28±3.02 & 97.33±5.54 & 42.65±9.27 & 37.74±8.37 \\
    & F1  & 94.42±0.99 & 92.82±1.25 & 98.79±0.08 & 78.04±3.16 & 95.78±8.46 & 37.51±6.75 & 28.85±7.00 \\
    \midrule

    \multirow{3}{*}{DH-SNN} 
    & Acc & 97.92±0.49 & 95.66±1.96 & 98.58±0.15 & 56.64±1.34 & 98.00±6.00  & 65.40±4.87 & 47.73±1.35
 \\
    & Pre & 97.38±1.13 & 95.65±1.96 & 98.50±0.15 & 54.09±14.92 & 98.75±3.75 & 40.62±11.71 & 31.97±6.80
\\
    & F1  & 96.78±0.77 & 95.77±1.87 & 98.58±0.15 & 47.74±8.53 & 97.62±7.14 & 37.72±8.05 & 33.46±6.85 \\
    \midrule

    \multirow{3}{*}{Rhythm-SNN} 
    & Acc & 97.90±0.69 & 95.81±1.16 & 98.45±0.14 & 88.57±2.41 & 98.33±5.00 & 65.57±5.60 & 50.90±2.04 \\
    & Pre & 97.10±1.34 & 95.79±1.17 & 98.47±0.14 & 88.37±2.64 & \bf{99.00±3.00} & 41.20±11.72 & 43.38±7.75 \\
    & F1  & 96.78±1.18 & 95.88±1.18 & 98.45±0.14 & 88.23±2.70 & 97.78±6.67
 & 36.72±9.06 & 39.75±2.36 \\
    \midrule
    \midrule
\multirow{3}{*}{\textbf{\modelname{}}} 
& Acc & \textbf{98.97±0.07}$^*$ & \textbf{96.98±0.03}$^*$ & \textbf{99.72±0.14}$^*$ & \textbf{88.72±2.09} & \textbf{98.33±5.00} & \textbf{67.62±5.21} & \textbf{64.65±1.79}$^*$ \\
& Pre & \textbf{98.43±0.05}$^*$ & \textbf{97.12±0.03}$^*$ & \textbf{99.72±0.14}$^*$ & \textbf{88.54±2.11} & \underline{{98.75±3.75}} & \textbf{49.30±8.83}$^*$ & \textbf{63.11±4.35}$^*$ \\
& F1  & \textbf{98.19±0.03}$^*$ & \textbf{96.96±0.03}$^*$ & \textbf{99.72±0.14}$^*$ & \textbf{88.48±2.26}$^*$ & \textbf{98.29±5.14}$^*$ & \textbf{41.51±7.19}$^*$ & \textbf{59.35±2.43}$^*$ \\
    \bottomrule
    \end{tabular}
    }
\end{minipage}
\hfill
\begin{minipage}[c]{0.31\textwidth}
    \centering
\includegraphics[width=0.96\textwidth]{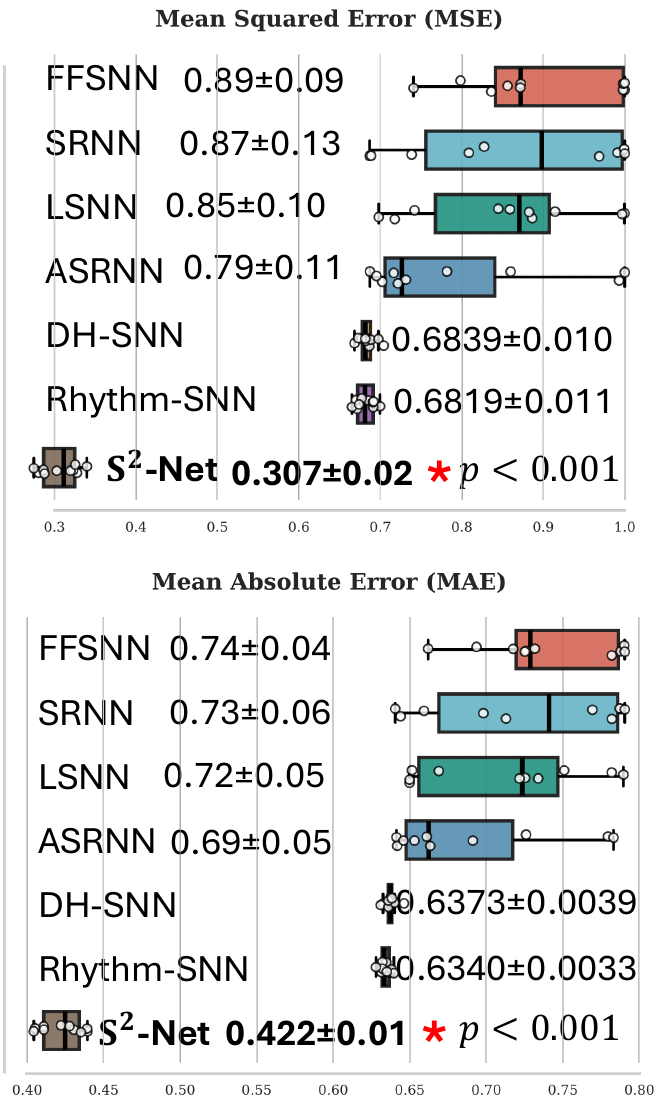}
    \vskip -0.1in
    \captionof{figure}{Brain decoding performance on HCP-WM dataset.}
    \label{fig:recon}
\end{minipage}
\vskip -0.15in
\end{table*}

\textbf{Firing dynamics across neurodegenerative conditions.} Moving beyond classification performance, we analyze the output spike ($\delta(t)$) to decode the latent spatiotemporal spiking patterns associated with each pathology.
In the ADNI cohort (Fig. \ref{fig:spike}, left), the AD group exhibited a consistently elevated mean firing rate compared to CN, suggesting a global state of network hyperexcitability. This aligns with established evidence that amyloid-$\beta$ accumulation disrupts inhibitory interneurons, leading to aberrant disinhibition in early disease stages \citep{palop2016network}. The recognized brain regions highlight that this hyperactivity is topologically concentrated in the medial temporal lobe and posterior hubs, corroborating the distinct, hypersynchronous spiking events observed in the raster plots. In contrast, the NIFD cohort (Fig. \ref{fig:spike}, middle) showes an overlapping global firing profile between FTD and CN. \modelname{} identified focal activation clusters in the frontal and temporal regions. This suggests that the signal loss from focal atrophy is dynamically counterbalanced by local circuit hyperexcitability in spared perilesional tissues, resulting in a net-neutral global output despite severe local pathology \citep{benussi2017transcranial}. PPMI cohort (Fig. \ref{fig:spike}, right) shows a rapid rise in firing rate, maintaining a higher steady-state amplitude than CN. Topological analysis revealed preferential recruitment of the basal ganglia and sensorimotor networks. The raster plots indicate sustained, high-frequency activity, which likely reflects the pathological oscillatory dynamics (e.g., beta-band hypersynchrony) characteristic of dopamine-depleted basal ganglia-thalamocortical loops \citep{ mcgregor2019circuit}. More detailed results are shown in Fig. \ref{fig:spike2} of Appendix.

\begin{figure*}
\centering
\includegraphics[width=0.99\linewidth, trim=0 1 0 0, clip]{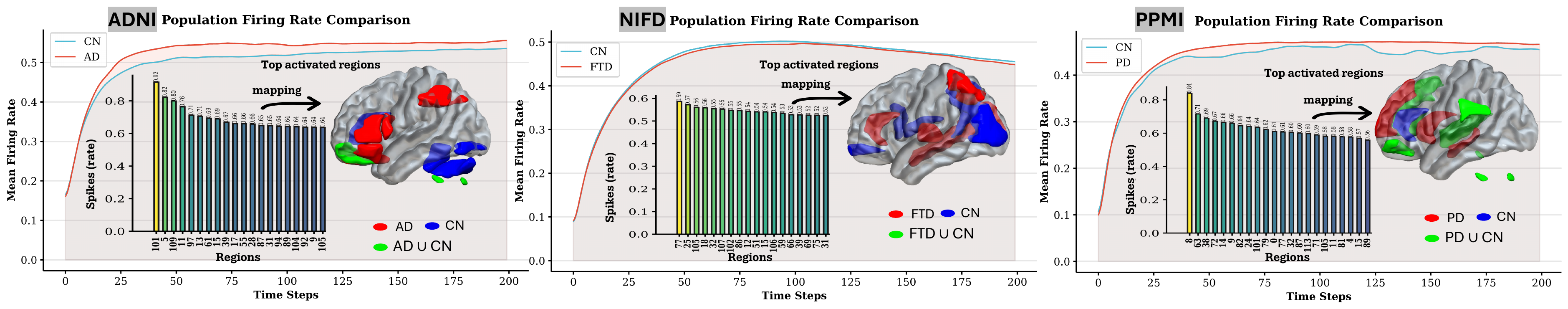}
\vskip -0.10in
\caption{\small Spatiotemporal spiking representations and disease-specific topological region identification. There are three dataset-specific columns: ADNI, NIFD, and PPMI, illustrating a top-down visual logic from macro-level dynamics to micro-level patterns: Top Row. Mean firing rate curves for healthy control (CN, blue) and disease (red) groups. Bottom Rows. Ranked brain region bar charts based on classification contribution (quantified by spiking intensity (rates)) and their corresponding anatomical mapping to 3D brain projections. }
\vskip -0.1in
\label{fig:spike}
\end{figure*}
\subsection{Versatility in General Temporal Processing}
Having demonstrated the framework's superior performance in interpreting complex brain dynamics, we further validate its generalization capability on standard temporal processing benchmarks lacking graph structures.
We adapt the \modelname{} to \textit{learn} the underlying connectivity directly from the data. Specifically, we treat the adjacency matrix (and the associated transport lags) as learnable latent parameters, allowing the model to self-organize the optimal synchronization topology for each task (Appendix \ref{app:stability}).
We conduct several experiments across five widely used benchmarks: Sequential MNIST (S-MNIST) \citep{le2015simple}, Permuted Sequential-MNIST (PS-MNIST) \citep{le2015simple}, Spiking Heidelberg Digits (SHD) \citep{cramer2020heidelberg}, Electrocardiogram (ECG) classification \citep{yin2021accurate}, and Google Speech Commands (GSC) \citep{warden2018speech} (Appendix \ref{dataset}). 

\begin{wraptable}{l}{0.45\linewidth} 
    \centering
    \caption{Performance comparison (accuracy \%) across datasets.}
    \label{tab:rhythm_snn_comparison_horizontal}
    \setlength{\tabcolsep}{2pt} 
    \renewcommand{\arraystretch}{1.05}
    \resizebox{1\linewidth}{!}{%
    \begin{tabular}{l c c c c c}
    \toprule
    \textbf{Model} & \textbf{S-M} & \textbf{PS-M} & \textbf{SHD} & \textbf{ECG} & \textbf{GSC} \\
    \midrule
    FFSNN      & 59.24 & 41.84 & 48.10 & 77.64 & 90.56 \\
    SRNN       & 81.51 & 64.89 & 81.60 & 71.56 & 91.76 \\
    LSNN       & 96.24 & 87.20 & 89.84 & 81.93 & 91.20 \\
    ASRNN      & 98.70 & 94.30 & 82.82 & 85.90 & 92.10 \\
    DH-SNN     & 98.90 & 94.52 & 89.84 & 86.35 & 94.05 \\
    Rhythm-SNN & 98.94 & 96.73 & 91.07 & 86.41 & 94.47 \\
    \midrule
    \textbf{\modelname{}} & \textbf{98.97} & \textbf{96.94} & \textbf{91.39} & \textbf{89.23} & \textbf{94.58} \\
    \bottomrule
    \end{tabular}}
\end{wraptable}

\textbf{Discussion.} In the main text, we solely compare our framework against SNN-series methods to ensure a fair evaluation of neuromorphic efficiency and biological plausibility. However, in order to demonstrate the broad applicability and generalization of \modelname{} beyond the neuromorphic domain, we have conducted extensive experiments against a diverse range of general deep learning methods (e.g., LSTMs, Transformers, TCN, Mamba \citep{gu2023mamba}, and synchronization-based models, i.e., AKOrN \citep{miyato2025artificial}, GraphCON \citep{rusch2022graph}, KuramotoGNN \citep{nguyen2024coupled} and SR-Net \citep{harikrishnan2021noise}). Please refer to Appendix \ref{real_app} for these detailed comparisons and full experimental results. The detailed ablation studies are shown in \ref{ablation_study}.

Table \ref{tab:rhythm_snn_comparison_horizontal} highlights the framework's robust generalization. First, validating \textit{temporal binding}, the model achieves 96.94\% on the challenging PS-MNIST task. This indicates a superior ability to solve the binding problem over long horizons. Second, in terms of \textit{energy-efficient signal processing}, the model sets a new SOTA on ECG (89.23\%) and remains highly competitive on the sparse neuromorphic SHD dataset. Finally, for \textit{semantic reasoning}, it achieves leading accuracy on GSC (94.58\%). These findings validate that the proposed mechanisms generalize beyond their original scope of brain dynamics to diverse data modalities, thereby highlighting that the principles governing brain dynamics can be translated to general-purpose computation.
\section{Conclusion}
We presented the spiking-by-synchronization neural network (\modelname{}), a biologically plausible architecture that unifies micro-scale spiking dynamics with macro-scale oscillatory synchronization. 
By modeling the brain's iterative bottom-up integration and top-down modulation via synchronization with phase lag, \modelname{} moves beyond simple rate coding to leverage the computational power of precise rhythmic timing. 
Extensive experiments show \modelname{}'s superiority across six human brain datasets and generalize robustly to standard temporal benchmarks, achieving SOTA performance. These results indicate that incorporating transient oscillatory mechanisms not only enhances biological realism for neuroscientific modeling but also provides a powerful control mechanism for efficient spatiotemporal information processing in broad neuromorphic applications.

\newpage
\appendix
\onecolumn
\section{Appendix}

\subsection{Dataset Descriptions}
\label{dataset}

In this section, we describe the datasets used to evaluate the proposed model across brain state recognition, disease diagnosis, temporal, auditory, neuromorphic, and sequence modeling tasks.


\textbf{1. The Lifespan Human Connectome Project Aging (HCP-A) dataset} \citep{bookheimer2019lifespan}.
The HCP-A dataset provides a comprehensive characterization of brain functional and structural organization across the aging process and is widely used for task recognition studies. It consists of data from 717 subjects, including both task-based fMRI scans (4,846 time series in total) and diffusion-weighted imaging (DWI) scans (717 in total). The dataset includes four task conditions related to memory and sensorimotor processing: VISMOTOR, CARIT, FACENAME, and Resting State. Each fMRI scan contains 300 time points. In our experiments, these task conditions are treated as distinct categories in a four-class classification problem.

\textbf{2. The Human Connectome Project – Young Adults (HCP-YA) dataset} \citep{van2013wu}.
The HCP-YA dataset contains high-quality task-based and resting-state fMRI data from 1,200 healthy young adults (each includes two scans). It includes seven cognitive tasks: Motor, Relational, Social, Working Memory, Language, Emotion, and Gambling. Each task-based fMRI scan consists of 175 time points. In addition, the Working Memory task (denoted as HCP-WM) alternates between 2-back and 0-back conditions with body, place, face, and tool stimuli, interleaved with fixation periods. A resting-state period follows every two task blocks, resulting in fMRI scans with 405 time points. Data preprocessing follows the pipeline described in \citep{dan2024exploring}.
For both HCP-A and HCP-YA datasets, the brain is parcellated into 116 regions using the AAL atlas \citep{tzourio2002automated}. Structural connectivity (SC) is represented as a $116 \times 116$ matrix, where each entry denotes the number of reconstructed white-matter fibers between two regions, normalized by the total fiber count of each subject. Functional connectivity (FC) is computed using Pearson correlation between regional BOLD time series.
In our experiments, the HCP-YA tasks are treated as a seven-class classification problem, while the HCP-WM conditions form a nine-class classification problem. We report results using 10-fold cross-validation.

\textbf{3. UK Biobank (UKB) dataset} \citep{miller2016multimodal}.
The UK Biobank is a large-scale population-based cohort that provides multimodal neuroimaging data from 16,240 participants. We use task-based fMRI data from the UKB imaging subset for task recognition. The dataset includes multiple cognitive paradigms designed to probe visual, motor, and higher-order cognitive functions. Each scan contains temporally aligned BOLD time series acquired under standardized imaging protocols. Similar to the HCP datasets, FC is computed from regional BOLD signals, and SC is derived when diffusion data are available. The UKB tasks are treated as distinct classes in a multi-class classification setting, enabling evaluation of model generalization on a large and heterogeneous population.


\textbf{4. Alzheimer’s Disease Neuroimaging Initiative (ADNI) dataset} \citep{jack2008alzheimer}.
The ADNI dataset is a longitudinal, multi-site study aimed at understanding the progression of Alzheimer’s disease (AD). We use resting-state fMRI and diffusion MRI data from 135 subjects spanning cognitively normal (CN), mild cognitive impairment (MCI), and Alzheimer’s disease (AD) groups. Functional connectivity is constructed from resting-state BOLD signals, while structural connectivity is estimated from diffusion tractography. The diagnostic labels are used to formulate a disease classification task, assessing the ability of the model to distinguish neurodegenerative stages.

\textbf{5. Parkinson’s Progression Markers Initiative (PPMI) dataset} \citep{marek2011parkinson}.
The PPMI dataset focuses on Parkinson’s disease (PD) and related disorders, providing multimodal imaging data including resting-state fMRI and diffusion MRI. A multi-center study that collects neuroimaging data from 173 subjects, including individuals with Parkinson’s disease (PD), scans without evidence of dopaminergic deficit (SWEDD), prodromal PD, and CN. Similar preprocessing and connectivity construction procedures are applied as in ADNI. The resulting FC and SC representations are used for binary disease diagnosis, evaluating the sensitivity of the proposed method to Parkinson’s-related network alterations.

\textbf{6. Neuroimaging in Frontotemporal Dementia (NIFD) dataset} \citep{planche2023anatomical}.
This dataset focuses on frontotemporal dementia (FTD) and includes resting-state fMRI data from 1,010 subjects. Participants are categorized into CN, logopenic variant of primary progressive aphasia (LVPPA), behavioral variant frontotemporal dementia (BV), progressive non-fluent aphasia (PNFA), and semantic variant (SV) groups. We use resting-state fMRI data to construct functional brain networks, capturing disease-related disruptions in large-scale connectivity patterns. The classification task is formulated as FTD versus control, providing an additional benchmark for evaluating model robustness across different neurodegenerative diseases.

\textbf{7. S-MNIST and PS-MNIST \citep{le2015simple}.}
Both datasets are derived from the original MNIST digit recognition benchmark by converting
each $28\times 28$ image into a temporal sequence of length 784 via raster scanning.  
For S-MNIST, pixels are fed in natural scan order, while PS-MNIST applies a fixed random permutation
to remove spatial locality. Pixel intensities are injected as input current to the encoding layer,
transforming non-spiking inputs into spike-based sequences for downstream SNN processing.

\textbf{8. SHD \citep{cramer2020heidelberg}.}
The Heidelberg Digits dataset contains $\sim$10,000 audio recordings of spoken digits (0–9)
in English and German, captured from 12 speakers. Each speaker utters each digit $\sim$40 times
across both languages. Audio is transformed into spike trains using a biologically inspired cochlea model.
Following standard practice, each sequence is represented as 1000 time steps.  
The dataset is split into 8156 training and 2264 test samples.

\textbf{9. ECG \citep{yin2021accurate}.}
This dataset consists of six types of ECG waveforms (P, PQ, QR, RS, ST, TP). We follow the
preprocessing procedure in prior work, applying a variant of the level-crossing encoding method to the
derivative of the normalized ECG signal. Each channel is converted into two spike trains capturing
value-increasing and value-decreasing events, respectively.

\textbf{10. GSC \citep{warden2018speech}.}
The Google Speech Commands dataset (GSC) contains 64,727 utterances from 1881 speakers.
Following standard configuration, we use 12 classes selected from the 35 available labels:
“Yes”, “No”, “Up”, “Down”, “Left”, “Right”, “On”, “Off”, “Stop”, “Go”, along with “Unknown” and “Silence”.
Audio preprocessing follows prior work, extracting log Mel filterbank features and derivatives. Unlike low-level signal classification, recognizing these specific commands requires mapping complex spatiotemporal patterns to unified semantic representations. Thus, GSC serves as a benchmark for \textit{semantic reasoning}, testing the framework's ability to decode communicative intent from raw sensory streams.

\subsection{Energy Interpretation and Lyapunov Stability Analysis}
\label{app:stability}
To theoretically validate the stability of our proposed graph-aware vector Kuramoto dynamics, we analyze the system through the lens of potential energy landscapes. Although the full dynamics (Eq.~\ref{eq:vector-kuramoto}) include a drift term $\omega_{i,k}$, the synchronization behavior is primarily governed by the interaction and forcing terms. We demonstrate that these terms can be derived as the gradient descent of a global Lyapunov function $V({\theta}; t)$.

\begin{theorem}
    Under the assumption of symmetric structural coupling ($A_{ij}=A_{ji}$) and antisymmetric phase lags ($\alpha_{ij}=-\alpha_{ji}$), the interaction and sensory drive components of the update rule follow the negative gradient of the following scalar potential field:
\begin{equation}
\begin{split} 
V({\theta}; t)
=
& \underbrace{-\frac{K}{2N}\sum_{d=1}^{D}\sum_{i=1}^{N}\sum_{j=1}^{N} A_{ij}\, \cos\bigl(\theta_{j,d} - \theta_{i,d} - \alpha_{ij}\bigr)}_{{V}_{\mathrm{sync}}: \text{ System dynamics}} \\
& \underbrace{- \sum_{d=1}^{D}\sum_{i=1}^{N}\kappa_{i,d} \cos\bigl(\gamma_i(t) - \theta_{i,d}\bigr)}_{{V}_{\mathrm{data}}: \text{ System control}}.
\end{split}
\label{eq:lyapunov2}
\raisetag{4em} 
\end{equation}
\end{theorem}
\begin{proof}
    Taking the partial derivative of $V$ with respect to the $k$-th phase component of node $i$, denoted $\theta_{i,d}$, we obtain:
\begin{equation}
\begin{aligned}
&-\frac{\partial V}{\partial \theta_{i,d}}
= \frac{\partial}{\partial \theta_{i,d}} \left[ \sum_{m=1}^{N}\kappa_{m,d} \cos(\gamma_m - \theta_{m,d}) \right] \\
& + \frac{K}{2N} \frac{\partial}{\partial \theta_{i,d}} \left[ \sum_{m,n} A_{mn} \cos(\theta_{n,d} - \theta_{m,d} - \alpha_{mn}) \right].
\end{aligned}
\end{equation}
Leveraging the symmetry of $A$ and antisymmetry of $\alpha$, the interaction derivative simplifies, yielding:
\begin{equation}
 -\frac{\partial V}{\partial \theta_{i,d}} = \kappa_{i,d}\sin(\gamma_i(t) - \theta_{i,d}) + \frac{K}{N}\sum_{j=1}^{N} A_{ij} \sin(\theta_{j,d} - \theta_{i,d} - \alpha_{ij}).
\end{equation}
This result exactly matches the update terms in our dynamic equation (Eq.~\ref{eq:vector-kuramoto}), confirming that the system trajectory satisfies $\dot{\theta}_{i,d} \propto -\nabla_{\theta_i} V$. Consequently, in the limit where the sensory drive $\gamma(t)$ varies slowly relative to the relaxation time of the oscillator network, the total energy decreases monotonically along the system's trajectory:
\begin{equation}
\frac{dV}{dt} = \sum_{i,d} \frac{\partial V}{\partial \theta_{i,d}} \dot{\theta}_{i,d} = - \sum_{i,d} \left( \frac{\partial V}{\partial \theta_{i,d}} \right)^2 \leq 0,
\label{dv}
\end{equation}
ensuring convergence to a stable phase-locked configuration.
\end{proof}

\textit{Physical Interpretation.} The potential $V$ embodies a fundamental competition between two operational objectives:

\textbf{Internal Coherence (${V}_{\mathrm{sync}}$):} The first term forms a ``topological prior.'' It is minimized when the phase difference $\Delta \theta_{ij}$ between connected regions aligns with the anatomically prescribed lag $\alpha_{ij}$. This forces the latent state to respect the brain's structural geometry, smoothing out noise and ensuring signal propagation follows valid white-matter pathways.

\textbf{Data Consistency (${V}_{\mathrm{data}}$):} The second term acts as a ``data anchor.'' Modeled as a sinusoidal potential well centered at $\gamma_i(t)$, it pulls the oscillator state toward the current sensory observation. The stiffness parameter $\kappa_{i,k}$ determines the confidence in the observation.

By minimizing this composite energy, the network dynamically finds the optimal trade-off, adhering to the brain's structural constraints while faithfully tracking the real-time brain dynamics.

\textbf{Remark: Learnable Topology and Physical Interpretation.} 
For general temporal datasets lacking predefined graphs, we treat the adjacency matrix $A \in \mathbb{R}^{N \times N}$, which governs the coupling strength between $N$ hidden neurons, as a learnable parameter initialized from a random distribution. Importantly, to preserve the theoretical guarantees of Lyapunov stability derived in Eq. (\ref{dv}), we enforce a symmetric constraint (i.e., $A_{ij} = A_{ji}$) during the forward pass. This is implemented by parameterizing the underlying weight matrix $W$ and computing $A = \sigma(W + W^\top)/2$, where $\sigma$ denotes the sigmoid function for normalization. 
This formulation ensures that the network self-organizes \textit{directed} traveling waves, essential for encoding temporal order, while strictly maintaining the gradient flow structure of the energy landscape. Through end-to-end backpropagation, the network self-organizes its topology, dynamically enhancing connections that correspond to significant temporal dependencies while suppressing irrelevant ones. This mechanism allows \modelname{} to discover a latent topology tailored to the input domain:
\begin{itemize}
    \item For \textbf{Visual tasks (e.g., S-MNIST, PS-MINST)}, the learned matrix reconstructs the latent \textit{spatial geometry}, effectively mapping implicit temporal delays (e.g., vertical pixel neighbors) into explicit structural couplings for synchronization.
    \item For \textbf{Audio/Speech tasks (e.g., SHD, GSC)}, the learned matrix captures \textit{spectro-temporal binding} and \textit{semantic synergy}. By learning symmetric couplings, the network identifies which frequency bands tend to co-activate (e.g., the harmonic structure of vowels), facilitating the integration of distributed acoustic cues into coherent phonemes or semantic commands.
\end{itemize}

The detailed interpretation of the involved datasets is summarized in Table \ref{tab:physical_meaning_A}.

\begin{table}[!ht]
    \centering
    \caption{Physical interpretation of the adjacency matrix $A$ and the role of synchronization across different benchmarks.}
    \label{tab:physical_meaning_A}
    \renewcommand{\arraystretch}{1.3} 
    \begin{tabularx}{\textwidth}{l l l X} 
        \toprule
        \textbf{Dataset} & \textbf{Input Domain} & \textbf{Physical Meaning of $A$} & \textbf{Function (Role of Synchronization)} \\
        \midrule
        \textbf{S/PS-MNIST} & Visual (Flattened) & {Spatial Geometry} & Reconstructs latent 2D spatial shapes from flattened 1D temporal sequences. \\
        \textbf{SHD} & Auditory (Spikes) & {Tonotopic Binding} & Integrates harmonics across frequency bands to form coherent auditory objects. \\
        \textbf{ECG} & Physiological (1D) & {Phase/Morphology} & Tracks the cyclic phase transitions of heartbeats and captures waveform morphology. \\
        \textbf{GSC} & Speech (Spectrogram) & {Semantic Synergy} & Binds distributed acoustic cues into unified phonemes and semantic words. \\
        \bottomrule
    \end{tabularx}
\end{table}

In parallel to learning the topology $A$, we also explicitly model the phase lag $\alpha_{ij}$ to capture directed temporal dynamics. 
To strictly adhere to the Lyapunov stability condition (which requires $\alpha_{ij} = -\alpha_{ji}$), we parameterize the directional bias as a learnable antisymmetric matrix. 
Specifically, we initialize a raw parameter matrix $W_\alpha$ and compute the directed component as $B =\text{tanh} (W_\alpha - W_\alpha^\top)$ to ensure bounded directionality. 
This learnable direction is then modulated by the magnitude of the learned connectivity $A$ (via the transport cost $C_{ij} \propto 1/A_{ij}$), resulting in a composite phase lag $\alpha_{ij} \propto B_{ij} \cdot C_{ij}$. 

\subsection{Deploy \modelname{} in Downstream Applications}\label{subsec:model}

We utilize three distinct heads to drive feature representation learning.

\textit{(1) Self-Supervised Reconstruction.} To preserve intrinsic information within the hidden representations, a decoder $f_{\mathrm{dec}}$ acts as a regularizer by reconstructing the original input from the hidden membrane potential $U(t)$ (representing sub-threshold dynamics):
\begin{equation}
\label{eq:recon_loss}
\hat{X}_t = f_{\mathrm{dec}}(U(t)), \quad \mathcal{L}_{\mathrm{rec}} = \sum_{t} \| X(t) - \hat{X}(t) \|_2^2 + \lambda \sum_{t} \| \hat{X}(t) - \hat{X}(t-1) \|_2^2
\end{equation}
where $\lambda$ controls the strength of the smoothness regularization.

\textit{(2) Dynamic State Recognition.} For the sequence labeling task, we require a prediction for state status $\hat{y}^{\mathrm{(state)}}(t)$ for every time step $t$. The accumulated membrane potential $\mathcal{H}(t)$ of the readout layer is mapped to class probabilities via the projection $f^1_{\mathrm{cls}}$:
\begin{equation}
\hat{y}^{(\mathrm{state})}(t) =f^1_{\mathrm{cls}}(\mathcal{H}(t)), \quad \hat{y}(t) \in \mathbb{R}^{C \times T},
\end{equation}
where $C$ is the number of dynamic states. 

\textit{(3) Classification.} For subject-level classification, we aggregate temporal information by pooling the accumulated logits $\mathcal{H}(t)$ across the sequence length $T$ and predict the label $\hat{y}^{(\mathrm{cls})}$ through project $f_{\mathrm{cls}}^2$:
\begin{equation}
\hat{y}^{(\mathrm{cls})} = f^2_{\mathrm{cls}}\left( \frac{1}{T}\sum_{t=1}^{T} \mathcal{H}(t) \right).
\end{equation}




\subsection{Additional Experimental Results
}
\subsubsection{Simulated Experiments}
\label{sim_app}
\textbf{Experimental setup.} To validate \textit{binding-by-synchrony} and \textit{traveling-wave propagation}, we constructed two synthetic environments (Fig. \ref{fig:sim}): (a) disjoint squares and (b) nested topological structures\footnote{In computational neuroscience, nested structures are regarded as a canonical benchmark \citep{roelfsema1998object,chen1982topological} for figure-ground segregation. They challenge the system to solve the 'binding problem' by integrating spatially separated segments into a coherent whole.}. Objects are initialized with distinct intensities ($1.0$ vs. $0.5$) to enable feature-driven repulsion via the graph coupling term $A_{ij}$. We applied Gaussian perturbation at varying levels: a baseline $\sigma=0.05$ for (a), and a severe noise regime $\sigma=0.2$ for (b) to test robustness under low signal-to-noise ratio.

\begin{figure}
\centering
\includegraphics[width=1\linewidth, trim=0 0 0in 0in, clip]{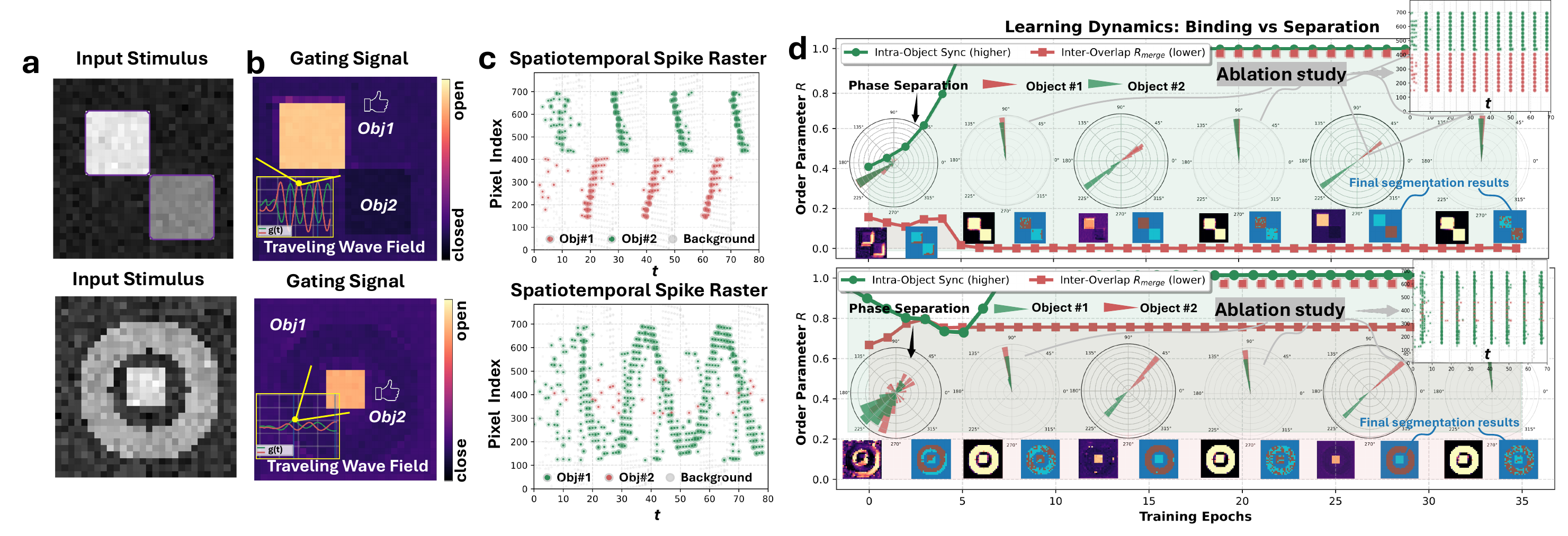}
\vskip -0.13in
\caption{\small Validation of the power of \modelname{} framework. (a) The input stimulus: disjoint object discovery (top) and nested topological binding (bottom). (b) The gating signal shows Object \#1 or (\#2) in an ``Open" phase (bright) while Object \#2 or (\#1) is ``Closed" (dark). Corresponding gating dynamics manifest as alternating sine waves, mathematically reproducing human attention rhythms (e.g., Theta oscillations). (c) Spatiotemporal spike raster: the framework leverages graph-structured delays to generate a traveling wave field, maintaining distinct temporal representations despite spatial entanglement. (d) Evolutionary curves of the order parameter $R$ quantify the learning process. Insets visualize the transition from initial phase chaos to stable separation (polar plots) and the final segmentation results.}
\label{fig:sim}
\end{figure}

\textbf{Evaluation metrics.}
We quantify binding quality and separation efficacy using Kuramoto-derived metrics (i.e., Kuramoto order parameter, KOP). Let $\mathcal{O}_q$ be the set of $M_q$ pixels for object $q$.
\textit{Intra-Object Synchronization ($R_{intra}$)} measures internal coherence: $R_{intra} = \frac{1}{Q} \sum_{q=1}^{Q} \left| \frac{1}{M_q} \sum_{j \in \mathcal{O}_q} e^{i\theta_j} \right|$. A value near $1.0$ indicates perfect binding.
\textit{Inter-Object Overlap ($R_{merge}$)} assesses desynchronization by treating the union of objects as a single population: $R_{merge} = \left| \frac{1}{\sum_{q=1}^Q M_q} \sum_{j \in \bigcup \mathcal{O}_q} e^{i\theta_j} \right|$. Ideally, $R_{merge} \to 0$ represents successful phase separation.

\textbf{Discrete object discovery and attention rhythms (Fig. \ref{fig:sim}-top)}
In the disjoint task (two objects shown in Fig. \ref{fig:sim}a), our \modelname{} breaks phase symmetry, driving the system from a disordered transient state into a stable anti-phase regime. Crucially, the gating dynamics $g(t)$ (Fig. \ref{fig:sim}b) emerge as smooth, alternating sinusoidal waves. This dynamic creates a ``push-pull" mechanism that mathematically replicates biological attention rhythms (e.g., Theta oscillations). The resulting spatiotemporal spike raster (Fig. \ref{fig:sim}c) reveals temporal chunking, where objects are phase-locked to alternating firing windows separated by refractory gaps.  As shown in Fig. \ref{fig:sim}d, $R_{intra}$ rapidly converges to $\approx 1.0$ while $R_{merge}$ collapses to $\approx 0.0$, confirming robust segregation.

\textbf{Topological complexity and perceptual bistability (Fig. \ref{fig:sim}-bottom).} Standard coupled oscillators often collapse into in-phase synchronization or a winner-take-all regime where the larger ring suppresses the smaller square (Fig. \ref{fig:sim}a). In contrast, \modelname{} achieves perceptual bistability by leveraging graph-aware traveling waves. As shown in the gating signal map (Fig. \ref{fig:sim}b), a phase gradient propagates along the ring, binding its spatially disjoint segments into a unified entity. This mechanism ensures balanced strength, despite the size imbalance, the smaller square is not suppressed but competes equally for representation. The spike raster (Fig. \ref{fig:sim}c) reveals a traveling wave topology, where the square fires precisely within the local phase troughs of the ring's propagation, resolving the topological conflict through autonomous temporal multiplexing.

\textbf{Ablation and dynamical stability}.
Upon turning off the graph-aware coupling and phase lag, the system immediately undergoes global synchronization collapse, with both $R_{intra}$ and $R_{merge}$ converging to $1.0$. This confirms that phase separation is not a trivial byproduct of the architecture but an emergent property of the oscillatory gating dynamics.

\subsubsection{Extended Baselines Experiments}
\label{real_app}
To evaluate the versatility and broad applicability of our proposed \modelname{} beyond the neuromorphic computing domain, we extend our comparative analysis to include a diverse suite of general-purpose and physics-inspired deep learning architectures. 

Specifically, we benchmark against three primary categories of advanced models:
\begin{itemize}
    \item \textbf{Standard Deep Sequence Models:} Transformer, LSTM, and Temporal Convolutional Networks (TCN).
    \item \textbf{State-Space Models:} Mamba, representing recent advancements in efficient long-range dependency modeling.
    \item \textbf{Physics-Inspired and Synchronization Baselines:} AKOrN, GraphCON, KuramotoGNN, and stochastic resonance (SR)-Net, representing state-of-the-art continuous-time and oscillator-based neural networks.
\end{itemize}

To ensure a rigorous and fair comparison, all baseline models are subjected to identical lightweight resource constraints: a hidden dimension of 64, a batch size of 16, and a maximum of 500 training epochs. We report the 10-fold cross-validation performance measured by Accuracy, Precision, and F1-score across six diverse neuroimaging datasets (HCPA, HCPYA, UKB, ADNI, PPMI, and NIFD).

As detailed in Table \ref{tab:expanded_baselines}, \modelname{} consistently outperforms both traditional sequence models and advanced synchronization-based networks across the majority of benchmarks. Notably, while standard architectures like Transformers and LSTMs suffer severe performance degradation on challenging or highly imbalanced datasets (e.g., ADNI and NIFD), our framework maintains robust predictive power. This superior performance underscores the efficacy of our time-delayed oscillatory modulation in extracting stable temporal representations under stringent resource constraints.

\begin{table}[ht]
\centering
\caption{Performance on eight models. 10-fold cross-validation results are reported for Accuracy, Precision, and F1-score (\%) across six brain-network datasets.}
\label{tab:expanded_baselines}
\resizebox{\linewidth}{!}{%
\begin{tabular}{llcccccc}
\toprule
\textbf{Model} & \textbf{Metric} & \textbf{HCPA} & \textbf{HCPYA} & \textbf{UKB} & \textbf{ADNI} & \textbf{PPMI} & \textbf{NIFD} \\
\midrule
\multirow{3}{*}{Transformer} 
& Acc & 96.89$\pm$0.39 & 96.48$\pm$0.22 & 98.83$\pm$0.26 & 78.17$\pm$3.76 & 66.10$\pm$2.68 & 61.11$\pm$2.03 \\
& Pre & 96.14$\pm$0.67 & 96.46$\pm$0.27 & 96.87$\pm$0.26 & 71.72$\pm$18.89 & 45.01$\pm$5.16 & 57.45$\pm$4.21 \\
& F1  & 95.85$\pm$0.63 & 96.47$\pm$0.24 & 96.81$\pm$0.26 & 64.44$\pm$12.65 & 41.79$\pm$1.76 & 54.17$\pm$3.58 \\
\midrule
\multirow{3}{*}{LSTM} 
& Acc & 97.26$\pm$0.70 & 95.39$\pm$0.24 & 98.95$\pm$0.24 & 74.78$\pm$6.62 & 66.15$\pm$4.72 & 62.74$\pm$2.30 \\
& Pre & 96.27$\pm$1.34 & 95.27$\pm$0.24 & 97.03$\pm$0.25 & 67.72$\pm$14.87 & 45.38$\pm$3.72 & 58.51$\pm$3.42 \\
& F1  & 96.27$\pm$0.90 & 95.28$\pm$0.24 & 99.94$\pm$0.24 & 63.62$\pm$12.31 & 44.82$\pm$3.34 & 55.38$\pm$3.63 \\
\midrule
\multirow{3}{*}{TCN} 
& Acc & 97.84$\pm$0.48 & 96.65$\pm$0.23 & 96.01$\pm$0.23 & 76.07$\pm$5.78 & 67.22$\pm$5.28 & 63.12$\pm$3.39 \\
& Pre & 97.21$\pm$0.60 & 96.62$\pm$0.28 & 96.08$\pm$0.22 & 63.56$\pm$3.89 & 40.63$\pm$12.19 & 62.39$\pm$1.87 \\
& F1  & 97.12$\pm$0.67 & 96.62$\pm$0.27 & 96.00$\pm$0.22 & 59.22$\pm$15.25 & 40.75$\pm$9.33 & 54.35$\pm$3.77 \\
\midrule
\multirow{3}{*}{Mamba} 
& Acc & 98.27$\pm$0.10 & 96.19$\pm$0.14 & 96.97$\pm$0.24 & 70.54$\pm$6.58 & 67.49$\pm$1.69 & 63.46$\pm$2.94 \\
& Pre & 98.42$\pm$0.03 & 96.15$\pm$0.16 & 97.03$\pm$0.23 & 52.98$\pm$19.41 & 47.46$\pm$3.51 & 60.51$\pm$4.04 \\
& F1  & 97.72$\pm$0.20 & 96.13$\pm$0.14 & 96.96$\pm$0.23 & 51.77$\pm$13.06 & 44.86$\pm$1.09 & 56.38$\pm$4.10 \\
\midrule
\multirow{3}{*}{AKOrN} 
& Acc & 94.34$\pm$0.76 & 86.31$\pm$1.71 & 92.22$\pm$1.68 & 85.27$\pm$7.52 & 66.82$\pm$4.87 & 58.79$\pm$2.68 \\
& Pre & 94.41$\pm$0.85 & 86.61$\pm$1.81 & 93.15$\pm$1.20 & 70.24$\pm$11.73 & 46.86$\pm$10.79 & 58.74$\pm$3.16 \\
& F1  & 94.30$\pm$0.81 & 86.23$\pm$1.73 & 91.88$\pm$1.65 & 76.46$\pm$10.28 & 39.81$\pm$8.82 & 58.31$\pm$3.13 \\
\midrule
\multirow{3}{*}{GraphCON} 
& Acc & 93.72$\pm$0.59 & 73.93$\pm$2.71 & 56.89$\pm$1.16 & 84.89$\pm$5.03 & 61.34$\pm$2.31 & 49.97$\pm$1.25 \\
& Pre & 93.61$\pm$0.56 & 74.13$\pm$2.41 & 57.35$\pm$1.53 & 74.46$\pm$10.68 & 47.06$\pm$8.98 & 33.26$\pm$7.36 \\
& F1  & 93.53$\pm$0.57 & 73.67$\pm$2.51 & 56.26$\pm$1.41 & 73.34$\pm$8.13 & 41.19$\pm$5.04 & 37.89$\pm$4.32 \\
\midrule
\multirow{3}{*}{KuramotoGNN} 
& Acc & 86.81$\pm$1.35 & 60.23$\pm$3.77 & 37.23$\pm$4.97 & 72.97$\pm$6.45 & 57.75$\pm$4.09 & 47.43$\pm$0.79 \\
& Pre & 87.56$\pm$0.79 & 61.97$\pm$2.68 & 40.09$\pm$1.47 & 73.32$\pm$6.53 & 47.89$\pm$13.21 & 44.14$\pm$2.24 \\
& F1  & 85.47$\pm$1.33 & 59.97$\pm$3.46 & 36.23$\pm$4.46 & 73.90$\pm$5.83 & 41.21$\pm$9.18 & 43.87$\pm$2.79 \\
\midrule
\multirow{3}{*}{SR-Net} 
& Acc & 60.41$\pm$2.12 & 36.79$\pm$3.18 & 59.24$\pm$1.40 & 71.81$\pm$3.27 & 57.41$\pm$2.50 & 29.37$\pm$2.63 \\
& Pre & 45.45$\pm$2.85 & 36.45$\pm$3.14 & 58.60$\pm$1.57 & 41.00$\pm$15.90 & 38.42$\pm$11.65 & 23.79$\pm$3.66 \\
& F1  & 47.25$\pm$2.35 & 35.63$\pm$3.16 & 58.57$\pm$1.54 & 43.17$\pm$4.98 & 35.88$\pm$8.39 & 19.13$\pm$1.88 \\
\midrule
\multirow{3}{*}{\textbf{\modelname{}}} 
& Acc & \textbf{98.97$\pm$0.07} & \textbf{96.98$\pm$0.03} & \textbf{99.72$\pm$0.14} & \textbf{98.33$\pm$5.00} & \textbf{67.62$\pm$5.21} & \textbf{64.45$\pm$1.79} \\
& Pre & \textbf{98.43$\pm$0.05} & \textbf{97.12$\pm$0.03} & \textbf{99.72$\pm$0.14} & \textbf{98.75$\pm$3.75} & \textbf{49.30$\pm$8.83} & \textbf{63.11$\pm$4.35} \\
& F1  & \textbf{98.19$\pm$0.03} & \textbf{96.96$\pm$0.03} & \textbf{99.72$\pm$0.14} & \textbf{98.29$\pm$5.14} & {41.51$\pm$7.19} & \textbf{59.35$\pm$2.43} \\
\bottomrule
\end{tabular}%
}
\end{table}

\begin{figure*}
\centering
\includegraphics[width=0.99\linewidth, trim=0 1 0 0, clip]{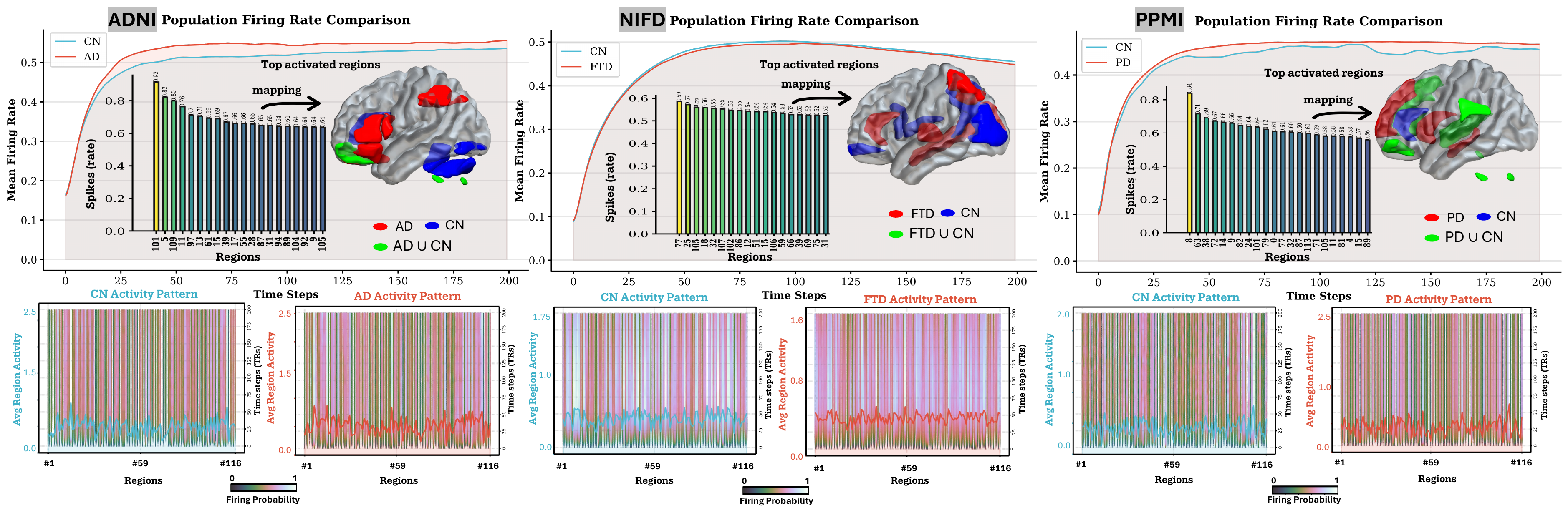}
\vskip -0.10in
\caption{\small Spatiotemporal spiking representations and disease-specific topological region identification. The figure is organized into three dataset-specific columns: ADNI, NIFD, and PPMI, illustrating a top-down visual logic from macro-level dynamics to micro-level patterns: Top Row. Mean firing rate curves for healthy control (CN, blue) and disease (red) groups. Middle Rows. Ranked brain region bar charts based on classification contribution (quantified by spiking intensity (rates)) and their corresponding anatomical mapping to 3D brain projections. Bottom Row. Spatiotemporal raster plots (brain regions vs. time steps) visualizing local spiking event density (firing probability from 0 to 1). }
\vskip -0.1in
\label{fig:spike2}
\end{figure*}

\subsection{Ablation Study}
\label{ablation_study}
\begin{figure}
\centering
\includegraphics[width=0.78\linewidth]{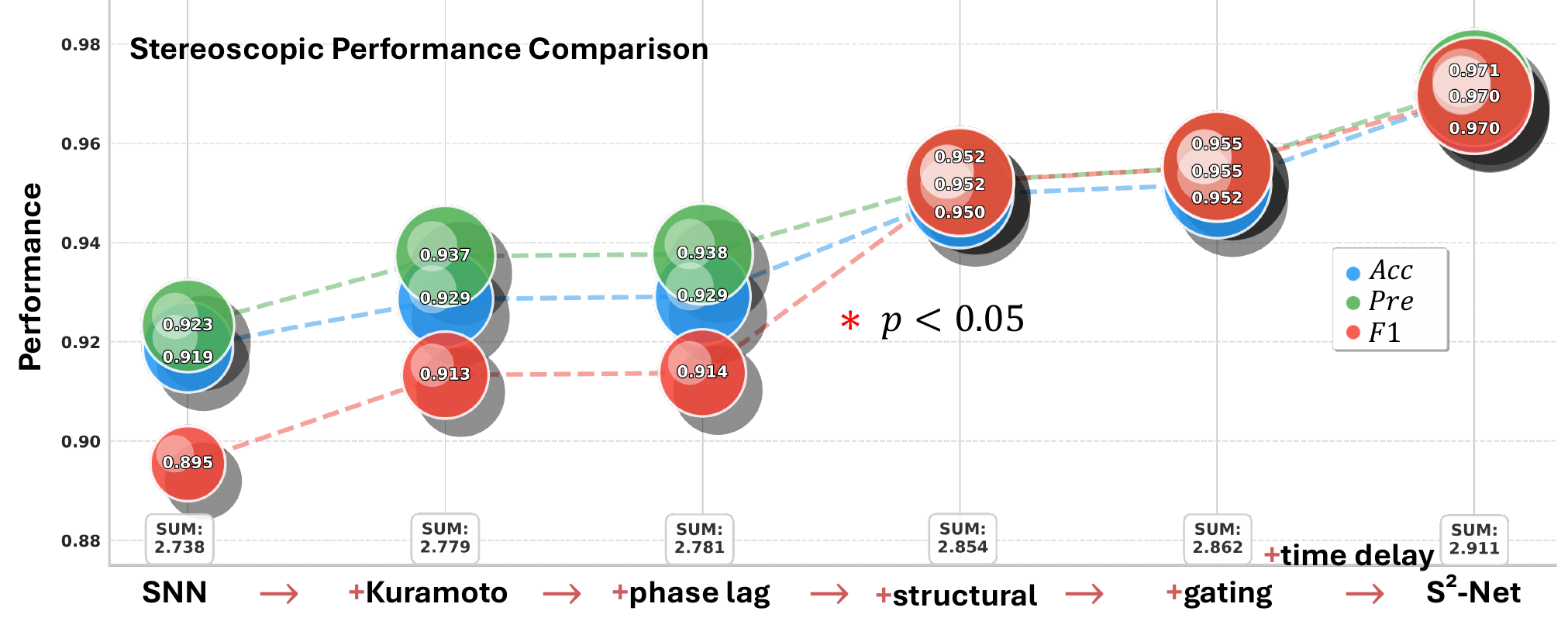}
 \vskip -0.1in
\caption{\small Ablation study results on HCP-YA dataset.}
\vskip -0.15in
\label{fig:ablation}
\end{figure}

\textbf{Contribution of proposed network components.} As shown in Fig. \ref{fig:ablation}, the ablation study reveals that our performance gains are not merely additive but stem from distinct dynamical properties introduced by each component.

\textit{(1) Establishing Basic Dynamics (+Kuramoto, +phase lag).} The initial integration of Kuramoto oscillators with phase lags transitions the system from static representations to traveling-wave dynamics (Sec. \ref{subsec:vector_kuramoto}). This enables the initial decoding of spatiotemporal depth, providing steady performance gains.

\textit{(2) Geometric Regularization via Lyapunov Stability.} The inclusion of structural information triggers a substantial performance leap (F1: $0.914 \to 0.950$), providing empirical validation for the stability analysis in Eq. (\ref{eq:lyapunov}). Without this term, phase differences $\theta_j - \theta_i$ evolve solely based on sensory drive $\gamma(t)$, potentially violating the underlying structural constraints. The introduction of the OT-derived lag $\alpha_{ij}$ modifies the energy landscape $V(\theta; t)$, penalizing phase relationships that contradict the geometric transport cost. Consequently, the network dynamics $\dot{\theta}_i$ are guided along a structure-preserving gradient flow (Eq. \ref{dv}). This confirms that embedding anatomical priors into the phase space acts as a powerful geometric regularizer, filtering out topologically inconsistent noise.

\textit{(3) Adaptive Filtering (+gating, +time delay).} The final evolution to the \modelname{} culminates in two decisive steps. First, the Gating mechanism (F1: $0.955$) serves as an adaptive filter, effectively suppressing residual fluctuations and improving the signal-to-noise ratio. Importantly, this clean signal state enables the final time delay module to operate effectively. By applying the time delay $\tau$, the system establishes high-dimensional synchronization, fully resolving complex spatiotemporal dependencies to secure the ultimate performance peak (F1: 0.970).

\textbf{Sensitivity Analysis of Hyperparameters.} 
To evaluate the robustness of \modelname{}, we conducted a sensitivity analysis on the PPMI dataset, focusing on three critical components of the oscillatory dynamics: coupling strength ($K$), time delay ($\tau$), and oscillation dimension ($D$). As reported in Table \ref{tab:sensitivity}, the default configuration ($K=1.0$, $\tau=2$, $D=4$) yields the most stable balance.

Specifically, extreme values of coupling strength ($K$) significantly impair the model's discriminative power. Weak coupling ($K=0.1$) fails to establish sufficient phase coherence, leading to the lowest F1-score ($38.12\%$), while overly strong coupling ($K=10$) forces premature global synchronization. We find that $K=1.0$ provides the necessary flexibility for phase separation, resulting in a substantial boost in precision ($49.30\%$) compared to the boundary cases.

The ablation of time delay ($\tau=0$) confirms that instantaneous coupling is suboptimal for resolving spatiotemporal patterns. However, we also observe that excessive delays (e.g., $\tau=5$) are equally detrimental, causing a drop in F1-score to $39.06\%$ due to 'over-frustration' in the oscillator network. This suggests that a moderate delay ($\tau=2$) is essential to break phase symmetry while maintaining the temporal alignment required for robust feature representation.

Regarding the \textit{oscillation dimension ($D$)}, the results in Table \ref{tab:sensitivity} demonstrate that \modelname{} is remarkably stable across different dimensions. The performance differences between $D=1, 4,$ and $8$ are minimal, with $D=1$ showing a slight edge in raw accuracy. We select $D=4$ as the default configuration as it provides a robust high-dimensional latent space that consistently captures long-range dependencies without the risk of dimensional collapse or overfitting.
\begin{table}[ht]
\centering
\caption{Sensitivity analysis of \modelname{} hyperparameters on the PPMI dataset. }
\label{tab:sensitivity}
\resizebox{0.65\linewidth}{!}{%
\begin{tabular}{lccc ccc}
\toprule
\textbf{Variable} & \bm{$K$} & \bm{$\tau$} & \bm{$D$} & \textbf{Acc (\%)} & \textbf{Pre (\%)} & \textbf{F1 (\%)} \\
\midrule
\multirow{2}{*}{$K$} 
& 0.1 & 2 & 4 & 66.77 $\pm$ 3.66 & 46.03 $\pm$ 8.22 & 38.12 $\pm$ 6.78 \\
& 10  & 2 & 4 & 65.92 $\pm$ 3.76 & 46.54 $\pm$ 3.64 & 40.24 $\pm$ 4.94 \\
\midrule
\multirow{2}{*}{$\tau$} 
& 1.0 & 0 & 4 & 67.28 $\pm$ 4.04 & 46.25 $\pm$ 3.89 & 40.34 $\pm$ 6.02 \\
& 1.0 & 5 & 4 & 66.93 $\pm$ 3.55 & 46.23 $\pm$ 8.32 & 39.06 $\pm$ 6.37 \\
\midrule
\multirow{2}{*}{$D$} 
& 1.0 & 2 & 1 & 67.76 $\pm$ 4.72 & 47.17 $\pm$ 4.61 & 41.86 $\pm$ 6.23 \\
& 1.0 & 2 & 8 & 67.43 $\pm$ 4.45 & 47.26 $\pm$ 3.68 & 39.44 $\pm$ 4.95 \\
\midrule
\textbf{Default} 
& \textbf{1.0} & \textbf{2} & \textbf{4} & \textbf{67.62 $\pm$ 5.21} & \textbf{49.30 $\pm$ 8.83} & \textbf{41.51 $\pm$ 7.19} \\
\bottomrule
\end{tabular}%
}
\end{table}

\subsection{Implementation Details}
\label{details}
We use the NVIDIA H100 NVL (8GPUs,94G/per) to conduct the involved experiments. We set time delay $\tau=2$, oscillation dimension $D=4$. The detailed hyperparameter settings are shown in Table \ref{tab:model_hyperparams}. Table \ref{tab:time} shows the number of parameters and inference time of each model, \modelname{} demonstrates decent parameter efficiency with a model size of only 0.22 MB, which is comparable to the lightweight Rhythm-SNN and significantly smaller than recurrent baselines like SRNN and LSNN. Regarding inference speed, we observe a trade-off. The inference throughput of \modelname{} is 2.79 samples/s.

\begin{table*}[htbp]
\centering
\caption{{Hyperparameter settings for different models.}}
\renewcommand{\arraystretch}{1.3}
\resizebox{0.9\textwidth}{!}{%
\begin{tabular}{l|ccccccc}
\hline
\textbf{Model} & FFSNN & SRNN & LSNN & ASRNN & DH-SNN & Rhythm-SNN & \textbf{\modelname{}} \\
\hline
\textbf{Optimizer} & Adam & Adam & Adam & Adam & Adam & Adam & Adam \\
\textbf{Learning rate} & $10^{-2}$ & $10^{-2}$ & $10^{-2}$ & $10^{-2}$ & $10^{-2}$ & $ 10^{-2}$ & $10^{-3}$ \\ 
\textbf{Batch size} & 256 & 64 & 256 & 64 & 16 & 16 & 16 \\
\textbf{Epochs} & 250 & 250 & 250 & 250 & 500 & 500 & 500 \\
\textbf{LR Schedule} & \makecell[c]{StepLR \\ \small{(100, 0.75)}} & \makecell[c]{StepLR \\ \small{(100, 0.75)}} & \makecell[c]{StepLR \\ \small{(100, 0.75)}} & \makecell[c]{StepLR \\ \small{(100, 0.75)}} & \makecell[c]{StepLR \\ \small{(100, 0.75)}} & \makecell[c]{StepLR \\ \small{(100, 0.75)}} & \makecell[c]{StepLR \\ \small{(100, 0.75)}} \\ 
\textbf{Hidden dim} & 1024 & 1024 & 1024 & 1024 & 1024 & 1024 & 64 \\
\hline
\end{tabular}}
\label{tab:model_hyperparams}
\end{table*}

\begin{table*}[htbp]
\centering
\caption{{Number of network parameters ans inference time on HCP-WM dataset ($N=116$ and $T=405$).}}
\renewcommand{\arraystretch}{1.3}
\resizebox{0.9\textwidth}{!}{%
\begin{tabular}{l|ccccccc} 
\hline
\textbf{Model} & FFSNN & SRNN & LSNN & ASRNN & DH-SNN & Rhythm-SNN & \textbf{\modelname{}} \\ 
\hline
\textbf{\# parameters (MB)} & 0.52 & 4.52 & 4.52 & 4.53 & 0.18 & 0.12 & \textbf{0.22} \\
\textbf{infer time (samples/s)} & 9.68 & 6.05 & 3.76 & 3.94 & 5.50 & 8.51 & \textbf{2.79} \\ 

\hline
\end{tabular}}
\label{tab:time}
\end{table*}

\textbf{Accessibility.}
All datasets used in this study are publicly available and can be accessed from their
respective sources. HCP-A is available at \url{https://www.humanconnectome.org/study/hcp-lifespan-aging/data-releases}. HCP-YA and HCP-WM are available at \url{https://www.humanconnectome.org/study/hcp-young-adult/data-releases}. UKB dataset is available at \url{https://www.ukbiobank.ac.uk/}. ADNI, PPMI and NIFD can be found in \url{https://ida.loni.usc.edu/login.jsp}. The S-MNIST and PS-MNIST datasets are derived from the MNIST dataset,
which is available at \url{http://yann.lecun.com/exdb/mnist/}. 
The SHD dataset can be accessed at \url{https://zenkelab.org/resources/spiking-heidelberg-datasets-shd/}.  
The ECG dataset is publicly available from PhysioNet at \url{https://physionet.org/content/qtdb/1.0.0/}.  
The Google Speech Commands (GSC) dataset can be obtained from
\url{https://www.tensorflow.org/datasets/catalog/speech_commands}.

All preprocessing steps and experimental protocols follow standard practices reported
in prior literature \citep{yan2025efficient}, ensuring full reproducibility of the results. 

\subsection{Limitations and Future Works}
\label{limitation}
While \modelname{} introduces biologically plausible mechanisms for binding, the model introduces specific hyperparameters that require calibration, namely the coupling strength $K$, time delay $\tau$, and the dimensionality of the oscillator space $D$. In biological systems, these constants naturally vary across brain regions and cognitive states. Finally, the lack of an adaptive mechanism to dynamically adjust $\tau$ for different datasets or signal modalities represents a limitation. Future iterations could incorporate learnable delays or delays based on structural tract lengths to improve generalization across diverse datasets.

\subsection{Impact Statement}
\label{impact}
This work introduces a biologically plausible (\modelname{}) inspired by cortical oscillatory synchronization. We anticipate several positive societal impacts stemming from this research. First, by advancing neuromorphic computing and leveraging time-delayed coordination, our model offers a pathway toward significantly reducing the carbon footprint and computational costs of AI systems, facilitating sustainable deployment on resource-constrained edge devices. Second, our framework demonstrates potential in analyzing complex spatiotemporal biomedical data, as evidenced by our experiments on neurodegenerative disease datasets (e.g., ADNI, PPMI, and NIFD). This could contribute to the development of robust, interpretable computer-aided diagnostic tools for the early detection of neurological disorders.

However, the analysis of sensitive neural and physiological data necessitates rigorous adherence to privacy standards and regulations to prevent data misuse. We encourage future researchers to prioritize fairness, interpretability, and data security when adapting this brain-inspired framework for real-world applications.



\begin{thebibliography}{47}
\providecommand{\natexlab}[1]{#1}
\providecommand{\url}[1]{\texttt{#1}}
\expandafter\ifx\csname urlstyle\endcsname\relax
  \providecommand{\doi}[1]{doi: #1}\else
  \providecommand{\doi}{doi: \begingroup \urlstyle{rm}\Url}\fi

\bibitem[Bellec et~al.(2018)Bellec, Salaj, Subramoney, Legenstein, and Maass]{bellec2018long}
Guillaume Bellec, Darjan Salaj, Anand Subramoney, Robert Legenstein, and Wolfgang Maass.
\newblock Long short-term memory and learning-to-learn in networks of spiking neurons.
\newblock \emph{Advances in Neural Information Processing Systems}, 31, 2018.

\bibitem[Benussi et~al.(2017)Benussi, Di~Lorenzo, Dell'Era, Cosseddu, Alberici, Caratozzolo, Cotelli, Micheli, Rozzini, Depari, et~al.]{benussi2017transcranial}
Alberto Benussi, Francesco Di~Lorenzo, Valentina Dell'Era, Maura Cosseddu, Antonella Alberici, Salvatore Caratozzolo, Maria~Sofia Cotelli, Anna Micheli, Luca Rozzini, Alessandro Depari, et~al.
\newblock Transcranial magnetic stimulation distinguishes alzheimer disease from frontotemporal dementia.
\newblock \emph{Neurology}, 89\penalty0 (7):\penalty0 665--672, 2017.

\bibitem[Bookheimer et~al.(2019)Bookheimer, Salat, Terpstra, Ances, Barch, Buckner, Burgess, Curtiss, Diaz-Santos, Elam, et~al.]{bookheimer2019lifespan}
Susan~Y Bookheimer, David~H Salat, Melissa Terpstra, Beau~M Ances, Deanna~M Barch, Randy~L Buckner, Gregory~C Burgess, Sandra~W Curtiss, Mirella Diaz-Santos, Jennifer~Stine Elam, et~al.
\newblock The lifespan human connectome project in aging: an overview.
\newblock \emph{Neuroimage}, 185:\penalty0 335--348, 2019.

\bibitem[Breakspear(2017)]{breakspear2017dynamic}
Michael Breakspear.
\newblock Dynamic models of large-scale brain activity.
\newblock \emph{Nature Neuroscience}, 20\penalty0 (3):\penalty0 340--352, 2017.

\bibitem[Bronski et~al.(2018)Bronski, Carty, and DeVille]{bronski2018configurational}
Jared~C Bronski, Thomas Carty, and Lee DeVille.
\newblock Configurational stability for the kuramoto--sakaguchi model.
\newblock \emph{Chaos: An Interdisciplinary Journal of Nonlinear Science}, 28\penalty0 (10), 2018.

\bibitem[Chen(1982)]{chen1982topological}
Lin Chen.
\newblock Topological structure in visual perception.
\newblock \emph{Science}, 218\penalty0 (4573):\penalty0 699--700, 1982.

\bibitem[Cramer et~al.(2020)Cramer, Stradmann, Schemmel, and Zenke]{cramer2020heidelberg}
Benjamin Cramer, Yannik Stradmann, Johannes Schemmel, and Friedemann Zenke.
\newblock The heidelberg spiking data sets for the systematic evaluation of spiking neural networks.
\newblock \emph{IEEE Transactions on Neural Networks and Learning Systems}, 33\penalty0 (7):\penalty0 2744--2757, 2020.

\bibitem[Cuturi(2013)]{cuturi2013sinkhorn}
Marco Cuturi.
\newblock Sinkhorn distances: Lightspeed computation of optimal transport.
\newblock \emph{Advances in neural information processing systems}, 26, 2013.

\bibitem[Dan et~al.(2024)Dan, Wei, Kim, and Wu]{dan2024exploring}
Tingting Dan, Ziquan Wei, Won~Hwa Kim, and Guorong Wu.
\newblock Exploring the enigma of neural dynamics through a scattering-transform mixer landscape for riemannian manifold.
\newblock In \emph{International Conference on Machine Learning}, pages 9976--9990. PMLR, 2024.

\bibitem[Dan et~al.(2025)Dan, Ding, and Wu]{dan2025explore}
Tingting Dan, Jiaqi Ding, and Guorong Wu.
\newblock Explore brain-inspired machine intelligence for connecting dots on graphs through holographic blueprint of oscillatory synchronization.
\newblock \emph{Nature Communications}, 16\penalty0 (1):\penalty0 9425, 2025.

\bibitem[Deco et~al.(2009)Deco, Jirsa, McIntosh, Sporns, and K{\"o}tter]{deco2009key}
Gustavo Deco, Viktor Jirsa, Anthony~R McIntosh, Olaf Sporns, and Rolf K{\"o}tter.
\newblock Key role of coupling, delay, and noise in resting brain fluctuations.
\newblock \emph{Proceedings of the National Academy of Sciences}, 106\penalty0 (25):\penalty0 10302--10307, 2009.

\bibitem[Felleman and Van~Essen(1991)]{felleman1991distributed}
Daniel~J Felleman and David~C Van~Essen.
\newblock Distributed hierarchical processing in the primate cerebral cortex.
\newblock \emph{Cerebral cortex (New York, NY: 1991)}, 1\penalty0 (1):\penalty0 1--47, 1991.

\bibitem[Fries(2005)]{fries2005mechanism}
Pascal Fries.
\newblock A mechanism for cognitive dynamics: neuronal communication through neuronal coherence.
\newblock \emph{Trends in cognitive sciences}, 9\penalty0 (10):\penalty0 474--480, 2005.

\bibitem[Greff et~al.(2020)Greff, Van~Steenkiste, and Schmidhuber]{greff2020binding}
Klaus Greff, Sjoerd Van~Steenkiste, and J{\"u}rgen Schmidhuber.
\newblock On the binding problem in artificial neural networks.
\newblock \emph{arXiv preprint arXiv:2012.05208}, 2020.

\bibitem[Gu and Dao(2023)]{gu2023mamba}
Albert Gu and Tri Dao.
\newblock Mamba: Linear-time sequence modeling with selective state spaces.
\newblock \emph{arXiv preprint arXiv:2312.00752}, 2023.

\bibitem[Harikrishnan and Nagaraj(2021)]{harikrishnan2021noise}
Nellippallil~Balakrishnan Harikrishnan and Nithin Nagaraj.
\newblock When noise meets chaos: Stochastic resonance in neurochaos learning.
\newblock \emph{Neural Networks}, 143:\penalty0 425--435, 2021.

\bibitem[Jack~Jr et~al.(2008)Jack~Jr, Bernstein, Fox, Thompson, Alexander, Harvey, Borowski, Britson, L.~Whitwell, Ward, et~al.]{jack2008alzheimer}
Clifford~R Jack~Jr, Matt~A Bernstein, Nick~C Fox, Paul Thompson, Gene Alexander, Danielle Harvey, Bret Borowski, Paula~J Britson, Jennifer L.~Whitwell, Chadwick Ward, et~al.
\newblock The alzheimer's disease neuroimaging initiative (adni): Mri methods.
\newblock \emph{Journal of Magnetic Resonance Imaging: An Official Journal of the International Society for Magnetic Resonance in Medicine}, 27\penalty0 (4):\penalty0 685--691, 2008.

\bibitem[Ju et~al.(2014)Ju, Dai, Cheng, and Yang]{ju2014dynamics}
Ping Ju, Qionglin Dai, Hongyan Cheng, and Junzhong Yang.
\newblock Dynamics in the sakaguchi-kuramoto model with two subpopulations.
\newblock \emph{Physical Review E}, 90\penalty0 (1):\penalty0 012903, 2014.

\bibitem[Kuramoto(2005)]{kuramoto2005self}
Yoshiki Kuramoto.
\newblock Self-entrainment of a population of coupled non-linear oscillators.
\newblock In \emph{International symposium on mathematical problems in theoretical physics: January 23--29, 1975, kyoto university, kyoto/Japan}, pages 420--422. Springer, 2005.

\bibitem[Le et~al.(2015)Le, Jaitly, and Hinton]{le2015simple}
Quoc~V Le, Navdeep Jaitly, and Geoffrey~E Hinton.
\newblock A simple way to initialize recurrent networks of rectified linear units.
\newblock \emph{arXiv preprint arXiv:1504.00941}, 2015.

\bibitem[Locatello et~al.(2020)Locatello, Weissenborn, Unterthiner, Mahendran, Heigold, Uszkoreit, Dosovitskiy, and Kipf]{locatello2020object}
Francesco Locatello, Dirk Weissenborn, Thomas Unterthiner, Aravindh Mahendran, Georg Heigold, Jakob Uszkoreit, Alexey Dosovitskiy, and Thomas Kipf.
\newblock Object-centric learning with slot attention.
\newblock \emph{Advances in neural information processing systems}, 33:\penalty0 11525--11538, 2020.

\bibitem[London and H{\"a}usser(2005)]{london2005dendritic}
Michael London and Michael H{\"a}usser.
\newblock Dendritic computation.
\newblock \emph{Annual Review of Neuroscience}, 28\penalty0 (1):\penalty0 503--532, 2005.

\bibitem[Maass(1997)]{maass1997networks}
Wolfgang Maass.
\newblock Networks of spiking neurons: the third generation of neural network models.
\newblock \emph{Neural networks}, 10\penalty0 (9):\penalty0 1659--1671, 1997.

\bibitem[Marek et~al.(2011)Marek, Jennings, Lasch, Siderowf, Tanner, Simuni, Coffey, Kieburtz, Flagg, Chowdhury, et~al.]{marek2011parkinson}
Kenneth Marek, Danna Jennings, Shirley Lasch, Andrew Siderowf, Caroline Tanner, Tanya Simuni, Chris Coffey, Karl Kieburtz, Emily Flagg, Sohini Chowdhury, et~al.
\newblock The parkinson progression marker initiative (ppmi).
\newblock \emph{Progress in Neurobiology}, 95\penalty0 (4):\penalty0 629--635, 2011.

\bibitem[McGregor and Nelson(2019)]{mcgregor2019circuit}
Matthew~M McGregor and Alexandra~B Nelson.
\newblock Circuit mechanisms of parkinson’s disease.
\newblock \emph{Neuron}, 101\penalty0 (6):\penalty0 1042--1056, 2019.

\bibitem[Miller et~al.(2016)Miller, Alfaro-Almagro, Bangerter, Thomas, Yacoub, Xu, Bartsch, Jbabdi, Sotiropoulos, Andersson, et~al.]{miller2016multimodal}
Karla~L Miller, Fidel Alfaro-Almagro, Neal~K Bangerter, David~L Thomas, Essa Yacoub, Junqian Xu, Andreas~J Bartsch, Saad Jbabdi, Stamatios~N Sotiropoulos, Jesper~LR Andersson, et~al.
\newblock Multimodal population brain imaging in the uk biobank prospective epidemiological study.
\newblock \emph{Nature Neuroscience}, 19\penalty0 (11):\penalty0 1523--1536, 2016.

\bibitem[Miyato et~al.(2025)Miyato, L{\"o}we, Geiger, and Welling]{miyato2025artificial}
Takeru Miyato, Sindy L{\"o}we, Andreas Geiger, and Max Welling.
\newblock Artificial kuramoto oscillatory neurons.
\newblock In \emph{The Thirteenth International Conference on Learning Representations}, 2025.

\bibitem[Neftci et~al.(2019)Neftci, Mostafa, and Zenke]{neftci2019surrogate}
Emre~O Neftci, Hesham Mostafa, and Friedemann Zenke.
\newblock Surrogate gradient learning in spiking neural networks: Bringing the power of gradient-based optimization to spiking neural networks.
\newblock \emph{IEEE Signal Processing Magazine}, 36\penalty0 (6):\penalty0 51--63, 2019.

\bibitem[Nguyen et~al.(2024)Nguyen, Honda, Sano, Nguyen, Nakamura, and Nguyen]{nguyen2024coupled}
Tuan Nguyen, Hirotada Honda, Takashi Sano, Vinh Nguyen, Shugo Nakamura, and Tan~Minh Nguyen.
\newblock From coupled oscillators to graph neural networks: Reducing over-smoothing via a kuramoto model-based approach.
\newblock In \emph{International Conference on Artificial Intelligence and Statistics}, pages 2710--2718. PMLR, 2024.

\bibitem[Palop and Mucke(2016)]{palop2016network}
Jorge~J Palop and Lennart Mucke.
\newblock Network abnormalities and interneuron dysfunction in alzheimer disease.
\newblock \emph{Nature Reviews Neuroscience}, 17\penalty0 (12):\penalty0 777--792, 2016.


\bibitem[Planche et~al.(2023)Planche, Mansencal, Manjon, Tourdias, Catheline, Coup{\'e}, Initiative, and the National Alzheimer's Coordinating Center~cohort]{planche2023anatomical}
Vincent Planche, Boris Mansencal, Jos{\'e}~V Manjon, Thomas Tourdias, Gwena{\"e}lle Catheline, Pierrick Coup{\'e}, {Frontotemporal Lobar Degeneration~Neuroimaging Initiative}, and {the National Alzheimer's Coordinating Center~cohort}.
\newblock Anatomical mri staging of frontotemporal dementia variants.
\newblock \emph{Alzheimer's \& Dementia}, 19\penalty0 (8):\penalty0 3283--3294, 2023.

\bibitem[Rao and Ballard(1999)]{rao1999predictive}
Rajesh~PN Rao and Dana~H Ballard.
\newblock Predictive coding in the visual cortex: a functional interpretation of some extra-classical receptive-field effects.
\newblock \emph{Nature neuroscience}, 2\penalty0 (1):\penalty0 79--87, 1999.

\bibitem[Roelfsema et~al.(1998)Roelfsema, Lamme, and Spekreijse]{roelfsema1998object}
Pieter~R Roelfsema, Victor~AF Lamme, and Henk Spekreijse.
\newblock Object-based attention in the primary visual cortex of the macaque monkey.
\newblock \emph{Nature}, 395\penalty0 (6700):\penalty0 376--381, 1998.

\bibitem[Roy et~al.(2019)Roy, Jaiswal, and Panda]{roy2019towards}
Kaushik Roy, Akhilesh Jaiswal, and Priyadarshini Panda.
\newblock Towards spike-based machine intelligence with neuromorphic computing.
\newblock \emph{Nature}, 575\penalty0 (7784):\penalty0 607--617, 2019.

\bibitem[Rusch et~al.(2022)Rusch, Chamberlain, Rowbottom, Mishra, and Bronstein]{rusch2022graph}
T~Konstantin Rusch, Ben Chamberlain, James Rowbottom, Siddhartha Mishra, and Michael Bronstein.
\newblock Graph-coupled oscillator networks.
\newblock In \emph{International conference on machine learning}, pages 18888--18909. PMLR, 2022.

\bibitem[Singer(1999)]{singer1999neuronal}
Wolf Singer.
\newblock Neuronal synchrony: a versatile code for the definition of relations?
\newblock \emph{Neuron}, 24\penalty0 (1):\penalty0 49--65, 1999.

\bibitem[Todri-Sanial et~al.(2024)Todri-Sanial, Delacour, Abernot, and Sabo]{todri2024computing}
Aida Todri-Sanial, Corentin Delacour, Madeleine Abernot, and Filip Sabo.
\newblock Computing with oscillators from theoretical underpinnings to applications and demonstrators.
\newblock \emph{Npj unconventional computing}, 1\penalty0 (1):\penalty0 14, 2024.

\bibitem[Tzourio-Mazoyer et~al.(2002)Tzourio-Mazoyer, Landeau, Papathanassiou, Crivello, Etard, Delcroix, Mazoyer, and Joliot]{tzourio2002automated}
Nathalie Tzourio-Mazoyer, Brigitte Landeau, Dimitri Papathanassiou, Fabrice Crivello, Octave Etard, Nicolas Delcroix, Bernard Mazoyer, and Marc Joliot.
\newblock Automated anatomical labeling of activations in spm using a macroscopic anatomical parcellation of the mni mri single-subject brain.
\newblock \emph{Neuroimage}, 15\penalty0 (1):\penalty0 273--289, 2002.

\bibitem[Van~Essen et~al.(2013)Van~Essen, Smith, Barch, Behrens, Yacoub, Ugurbil, Consortium, et~al.]{van2013wu}
David~C Van~Essen, Stephen~M Smith, Deanna~M Barch, Timothy~EJ Behrens, Essa Yacoub, Kamil Ugurbil, Wu-Minn~HCP Consortium, et~al.
\newblock The wu-minn human connectome project: an overview.
\newblock \emph{Neuroimage}, 80:\penalty0 62--79, 2013.

\bibitem[Varela et~al.(2001)Varela, Lachaux, Rodriguez, and Martinerie]{varela2001brainweb}
Francisco Varela, Jean-Philippe Lachaux, Eugenio Rodriguez, and Jacques Martinerie.
\newblock The brainweb: phase synchronization and large-scale integration.
\newblock \emph{Nature reviews neuroscience}, 2\penalty0 (4):\penalty0 229--239, 2001.

\bibitem[Warden(2018)]{warden2018speech}
Pete Warden.
\newblock Speech commands: A dataset for limited-vocabulary speech recognition.
\newblock \emph{arXiv preprint arXiv:1804.03209}, 2018.

\bibitem[Wu et~al.(2019)Wu, Deng, Li, Zhu, Xie, and Shi]{wu2019direct}
Yujie Wu, Lei Deng, Guoqi Li, Jun Zhu, Yuan Xie, and Luping Shi.
\newblock Direct training for spiking neural networks: Faster, larger, better.
\newblock In \emph{Proceedings of the AAAI Conference on Artificial Intelligence}, volume~33, pages 1311--1318, 2019.

\bibitem[Yan et~al.(2025)Yan, Yang, Wu, Liu, Zhang, Li, Tan, and Wu]{yan2025efficient}
Yinsong Yan, Qu~Yang, Yujie Wu, Hanwen Liu, Malu Zhang, Haizhou Li, Kay~Chen Tan, and Jibin Wu.
\newblock Efficient and robust temporal processing with neural oscillations modulated spiking neural networks.
\newblock \emph{Nature Communications}, 16\penalty0 (1):\penalty0 8651, 2025.

\bibitem[Yeung and Strogatz(1999)]{yeung1999time}
MK~Stephen Yeung and Steven~H Strogatz.
\newblock Time delay in the kuramoto model of coupled oscillators.
\newblock \emph{Physical review letters}, 82\penalty0 (3):\penalty0 648, 1999.

\bibitem[Yin et~al.(2020)Yin, Corradi, and Boht{\'e}]{yin2020effective}
Bojian Yin, Federico Corradi, and Sander~M Boht{\'e}.
\newblock Effective and efficient computation with multiple-timescale spiking recurrent neural networks.
\newblock In \emph{International Conference on Neuromorphic Systems 2020}, pages 1--8, 2020.

\bibitem[Yin et~al.(2021)Yin, Corradi, and Boht{\'e}]{yin2021accurate}
Bojian Yin, Federico Corradi, and Sander~M Boht{\'e}.
\newblock Accurate and efficient time-domain classification with adaptive spiking recurrent neural networks.
\newblock \emph{Nature Machine Intelligence}, 3\penalty0 (10):\penalty0 905--913, 2021.

\bibitem[Zenke and Vogels(2021)]{zenke2021remarkable}
Friedemann Zenke and Tim~P Vogels.
\newblock The remarkable robustness of surrogate gradient learning for instilling complex function in spiking neural networks.
\newblock \emph{Neural Computation}, 33\penalty0 (4):\penalty0 899--925, 2021.

\end{thebibliography}
\end{document}